\documentstyle{article}

\newsavebox{\theorembox}
\newsavebox{\lemmabox}
\newsavebox{\corollarybox}
\newsavebox{\propositionbox}
\newsavebox{\examplebox}
\newsavebox{\conjecturebox}
\newsavebox{\algbox}
\newsavebox{\qbox}
\newsavebox{\problembox}
\newsavebox{\definitionbox}
\newsavebox{\assumptionbox}
\newsavebox{\hypothesisbox}
\savebox{\theorembox}{\noindent\bf Theorem}
\savebox{\lemmabox}{\noindent\bf Lemma}
\savebox{\corollarybox}{\noindent\bf Corollary}
\savebox{\propositionbox}{\noindent\bf Proposition}
\savebox{\examplebox}{\noindent\bf Example}
\savebox{\conjecturebox}{\noindent\bf Conjecture}
\savebox{\algbox}{\noindent\bf Algorithm}
\savebox{\qbox}{\noindent\bf Question}
\savebox{\definitionbox}{\noindent\bf Definition}
\savebox{\problembox}{\noindent\bf Problem}
\savebox{\assumptionbox}{\noindent\bf Assumption}
\savebox{\hypothesisbox}{\noindent\bf Hypothesis}
\newtheorem{theorem}{\usebox{\theorembox}}
\newtheorem{lemma}[theorem]{\usebox{\lemmabox}}
\newtheorem{corollary}[theorem]{\usebox{\corollarybox}}
\newtheorem{proposition}[theorem]{\usebox{\propositionbox}}

\newtheorem{definition}{\usebox{\definitionbox}}

\newcommand{\qed}{\;\;\;\Box}
\newenvironment{proof}{\par{\bf Proof:}}{\(\qed\) \par}

\begin{document}

\title{Contagious Sets in Expanders}

\author{Amin Coja-Oghlan\thanks{Goethe University. {\tt acoghlan@math.uni-frankfurt.de}. Supported by ERC Starting
Grant 278857–PTCC (FP7). }
\and Uriel Feige\thanks{The Weizmann Institute. {\tt
    uriel.feige@weizmann.ac.il}. Supported in part by The Israel Science Foundation (grant No. 621/12) and by the Citi Foundation}
\and Michael Krivelevich\thanks{Tel-Aviv University. {\tt krivelev@post.tau.ac.il }.
    Research supported in part by: USA-Israel BSF Grant 2010115 and by grant 912/12 from the Israel Science
Foundation.}
        \and Daniel Reichman\thanks{The Weizmann Institute. {\tt
    daniel.reichman@gmail.com}. Supported in part by The Israel Science Foundation (grant No. 621/12) and by the Citi Foundation and by a Weizmann-Warwick Making Connections Grant: "The Interplay between Algorithms and Randomness."}}
%\author{Amin Coja Oghlan, Uriel Feige, Michael Krivelevich and Daniel Reichman}
%\institute{
%Department of Computer Science, The Weizmann Institute, Rehovot, Israel\\
%Department of Computer Science and Applied Mathematics,\\
%The Weizmann Institute of Science, Rehovot, Israel\\
%\email{uriel.feige@weizmann.ac.il;daniel.reichman@gmail.com}}

\maketitle

\begin{abstract}
We consider the following activation process in undirected graphs: a vertex is active either if it belongs to a set of initially activated vertices or if at some point it has at least $r$ active neighbors, where $r>1$ is the activation threshold.

A \emph{contagious set} is a set whose activation results with the entire graph being active. Given a graph $G$, let $m(G,r)$ be the minimal size of a contagious set. Computing $m(G,r)$ is NP-hard.

It is known that for every $d$-regular or nearly $d$-regular graph on $n$ vertices, $m(G,r) \le O(\frac{nr}{d})$.
We consider such graphs that additionally have expansion properties, parameterized by the spectral gap and/or the girth of the graphs.

The general flavor of our results is that sufficiently strong expansion (e.g., $\lambda(G)=O(\sqrt{d})$, or girth $\Omega(\log \log d)$) implies that $m(G,2) \le O(\frac{n}{d^2})$ (and more generally, $m(G,r) \le O(\frac{n}{d^{r/(r-1)}})$). Significantly weaker expansion properties suffice in order to imply that $m(G,2)\le O(\frac{n \log d}{d^2})$. For example, we show this for graphs of girth at least~7, and for graphs with $\lambda(G)<(1-\epsilon)d$, provided the graph has no 4-cycles. Nearly $d$-regular expander graphs can be obtained by considering the binomial random graph $G(n,p)$ with $p \simeq \frac{d}{n}$ and $d > \log n$. For such graphs we prove that $\Omega(\frac{n}{d^2 \log d}) \le m(G,2) \le O(\frac{n\log\log d}{d^2\log d})$ almost surely.

Our results are algorithmic, entailing simple and efficient algorithms for selecting contagious sets.

\end{abstract}
\newpage
\tableofcontents
\newpage
\section{Introduction}
Threshold models in graphs and networks have received much attention in diverse research fields. Typically in such models there is an undirected graph $G=(V,E)$ where every node $v \in V$ has a threshold function $t(v)$. In addition, it is assumed that every node can be in two states: either active or inactive. An initial set of nodes (termed seeds) is activated. An inactive vertex $v$ becomes active once it has at least $t(v)$ active neighbors. In this work we focus on progressive models: once a vertex is active, it remains active forever.

Threshold models emerge in various settings such as brain modeling, diffusion of innovation, ideas, and trends in social networks, resilience to cascading failures in financial networks, power grids and communication networks \cite{blume, Kleinberg1,Gra,peleg,Tlusty}. Within computer science, the rising popularity of social media has resulted in much interest in various optimization problems related to cascading behavior in networks \cite{Domingos,Kleinberg1,Mossel}.

We shall focus on threshold models where every vertex has the same threshold $r$ (we will mostly assume $r$ is small, e.g., $2$ or $3$). Such activation rules, which are often referred to as \emph{bootstrap percolation}, have been introduced in statistical physics settings \cite{Chal}. (A note regarding terminology. The term {\em bootstrap percolation} is sometimes used with the implicit assumption that the set of seeds is random. In this paper we use this term also when the set of seeds is selected deterministically rather than at random.)
Formally, in $r$-\emph{neighbor bootstrap percolation} we are given an undirected graph $G=(V,E)$ and an integer $r>1$. Every vertex is either \emph{active} or \emph{inactive}. A set of vertices composed entirely of active vertices is called active. Initially, a set of vertices $A_0$ is activated. These vertices are called \emph{seeds}. A contagious process evolves in discrete steps where for $i>0$,
$$A_i=A_{i-1}\cup \{v:|N(v)\cap A_{i-1}|\geq r\},$$
where $N(v)$ is the set of neighbors of $v$. In words, a vertex becomes active in a given step if it has at least $r$ active neighbors. We refer to $r$ as the \emph{threshold}. Set
$$\langle A_0 \rangle=\bigcup_iA_i.$$
\begin{definition}
Given $G=(V,E)$, a set $A_0\subseteq V$ is called \emph{contagious} if $\langle A_0 \rangle=V$. In words, activating $A_0$ results with the entire graph being activated.
The minimal cardinality of a contagious set is denoted by $m(G,r)$. For a contagious set $A_0$, the number of generations is the minimal integer $t$ with $\bigcup_{i \leq t}A_i=V$.
\end{definition}

Bootstrap percolation has been subjected to extensive research in computer science (see for example \cite{Ackerman,Chen09,rappaport}) as well as in probabilistic and combinatorial settings \cite{BB,BP,Lattice,peres,JanLuc}. It is known that in every $d$-regular graph $m(G,r) \le \frac{rn}{d+1}$ \cite{Ackerman,Reichman}. For certain families of graphs (a collection of disjoint cliques each of size $d+1$), $m(G,r)=\frac{rn}{d+1}$.

\subsection{Contagious sets in expander graphs: motivation}

In this work we study how $m(G,r)$ depends on the \emph{expansion} properties of $G$. Let $G$ be a $d$-regular graph. We shall distinguish between two types of expansion properties, and associate one parameter with each type. One type is what we refer to as {\em global expansion}. The parameter that we associate with it is $\lambda(G)$, the second largest eigenvalue (in absolute value) of the adjacency matrix of $G$. We focus on \emph{spectral expanders}, namely, graphs for which $\lambda(G) \le \delta d$ for some $\delta <1$ (observe that for every $d$-regular graph $\lambda(G) \leq d$). We refer to this class of graphs as $(n,d,\lambda)$-graphs, where $n$ is the number of vertices. The other type is what we refer to as {\em local expansion}. The parameter that we associate with it is the girth $g$ (the length of a shortest cycle in $G$). If $g \ge 2k+1$ this implies that every vertex has $d(d-1)^{k-1}$ distinct neighbors at distance $k$ from it.  We remark that large girth does not imply small $\lambda$ (a graph might have high girth without even being connected, in which case $\lambda = d$), and $\lambda<\delta d$ need not imply high girth (a graph with $\lambda<\delta d$ may have triangles and four-cycles). We also remark that our results concerning high girth graphs can be extended to graphs that do have short cycles, provided that every small set of vertices has a sufficiently large neighborhood. Details of this are omitted from this manuscript.

Expanders are rich mathematical objects with diverse applications in algebra, combinatorics, probability and theoretical computer science \cite{linial}. Furthermore, expander graphs are used in designing fault tolerant networks, hence it makes sense to study various algorithmic problems on expanders and there are several works in this flavor \cite{Broder,Capalbo,Kleinbergr}. Understanding optimization problems on expanders and random graphs may be useful in understanding these problems in worst-case settings (see for example \cite{Arora}).
The study of combinatorial optimization problems on graphs with high girth is quite natural as well.

Several works have demonstrated that expanders are resilient to random or adversarial faults in the sense that they keep a certain degree of connectivity in the presence of faulty edges or nodes \cite{AlonChung,Kaplan,eppstein}.
%The use of such fault-tolerant networks was advocated for network of processors as the preservation of a large connected components in the presence of faults is often crucial for keeping the functionality of the network \cite{AlonChung,Friedman}. Resilience of networks to cascading failures is also of interest (see for example \cite{blume}, although their threshold model is different from ours).
Our results imply that for expander graphs, $m(G,r)$ is substantially smaller than the bound $\frac{rn}{d+1}$ which holds for arbitrary $d$-regular graphs, especially when $d$ is large and $r$ is small. In fact, even relatively modest conditions on the girth of $G$ (e.g., girth larger than four) already entail upper bounds on $m(G,r)$ which are substantially smaller than $\frac{rn}{d+1}$. Hence properties (such as expansion) that make a network more resilient to \emph{static} failures might make it more vulnerable to cascading faults (within the bootstrap percolation model).

\subsection{Our results}

It will be convenient for us to distinguish between three algorithms for selecting seeds.

{\em Random-parallel.} In this algorithm one fixes a parameter $p \in (0,1)$ (that may depend on the input graph $G$), and initially activates each vertex independently with probability $p$. If the set of seeds (initially activated vertices) happens to be contagious the algorithm succeeds, and if not it fails. This is typically the algorithm implicitly associated with the term bootstrap percolation.

{\em Random-sequential.} This algorithm proceeds in rounds. In each round, the algorithm picks a new vertex at random to become a seed, but only among those vertices that have not been activated in previous rounds (neither by becoming seeds, nor by a cascade effect).

{\em Greedy.} This is a family of algorithms, parameterized by the greedy rule that is used. The algorithm proceeds in rounds. In each round the algorithm selects one vertex as a seed according to some greedy rule. A natural rule is to select the vertex whose activation will result in the largest cascade of newly activated vertices. In our work we shall consider other greedy rules as well.

Our first result concerns spectral expanders.
To put the following theorem in context one should note that for every $d$-regular graph $\lambda \ge \Omega(\sqrt{d})$, and that for most $d$-regular graphs $\lambda \le O(\sqrt{d})$ (see~\cite{linial}, for example).

\begin{theorem}
\label{thm:Expander}
Let $G$ be an $(n,d,\lambda)$-graph. If $\lambda=O(\sqrt{d})$ then $m(G,2)=O(\frac{n}{d^2})$. More generally, if $\lambda \le \frac{1}{\sqrt{l}}d$ and $l$ is sufficiently large,
then $m(G,2)=O(\frac{n}{l^2})$. Moreover, a contagious set can be chosen by the random-parallel algorithm (with a value of $p = O(l^{-2})$). For the randomly constructed contagious set, the number of generations until complete activation is $O(\log_l \log n+\log\log d)$ with probability $1-o(1)$.
\end{theorem}

Our next result concerns high girth graphs. The random-parallel algorithm is inappropriate in this case (for example, when the graph is composed of many separate components, $p$ might need to be very close to~1 to ensure that each component has at least two seeds), and hence we revert to the random-serial algorithm.

\begin{theorem}\label{thm:girth}
Let $G$ be a $d$-regular graph of girth at least $2k+1$. If $k \ge \log\log d$ then $m(G,2) = O(\frac{n}{d^2})$, and if $k < \log\log d$ then $m(G,2)=O(d^{\zeta}\frac{n}{d^2})$, where $\zeta = \frac{1}{2^{k-1}}$. Moreover, the contagious set can be chosen by the random-serial algorithm, in which case the number of generations until complete activation can be made at most $k$.
\end{theorem}

Proposition~\ref{pro:generations} shows that the number of generations in Theorem~\ref{thm:Expander} is best possible up to constant factors for random parallel activation, and Theorem~\ref{thm:girth} gives examples where random sequential activation leads to fewer generations than random parallel activation.

\begin{proposition}
\label{pro:generations}
For every $d$-regular graph, if every vertex is initially activated independently with probability at most $1/4$, then with probability $1-o(1)$ the number of generations until complete activation is at least $\log_d \log n$.
\end{proposition}

Theorems~\ref{thm:Expander} and~\ref{thm:girth} give nearly best possible bounds for $m(G,2)$ when $\lambda \le O(\sqrt{d})$ or the girth exceeds $2\log\log d$.

\begin{theorem}
\label{thm:LB}
Let $\epsilon > 0$ be an arbitrarily small positive constant. Then for $d$ large enough there are $(n,d,\lambda)$-graphs with $\lambda = O(\sqrt{d})$, girth $\Omega(\log \log d)$ and $m(G,2) \ge \Omega(\frac{n}{d^{2+\epsilon}})$.
\end{theorem}

The upper and lower bounds above extend to activation thresholds $r > 2$, with the adjustment that the terms $d^2$ need to be replaced by $d^{\frac{r}{r-1}}$ (for example, an upper bound of $m(G,2) \le O(\frac{n}{d^2})$ is replaced by $m(G,r) \le O(\frac{n}{d^{\frac{r}{r-1}}})$). See Section \ref{sec:bigthreshold} for precise statements of these results.

The upper bounds in Theorems~\ref{thm:Expander} and~\ref{thm:girth} are not known to be tight when $\lambda(G)$ approaches $d$ or when the girth approaches (from above) 5. In fact, we believe that they are not tight. One may conjecture that for every $\delta < 1$, an $(n,d,\lambda)$-graph with $\lambda < \delta d$ has $m(G,2) \le O(\frac{n}{d^2})$ (with the hidden constant in the $O$ notation depending on $\delta$). We do not know if this conjecture is true, but we do know that the bounds in Theorem~\ref{thm:Expander} are far from tight when $\lambda$ is fairly large.

\begin{proposition}
\label{pro:general_bound}
Let $G$ be an $(n,d,\lambda)$-graph where $\lambda<\delta d$ where $\delta<1$ is independent of $d$. Then there is a contagious set in $G$ of size $O(\frac{n}{d^{\frac{3}{2}}})$.  Moreover, the contagious set can be chosen by the random-parallel algorithm.
\end{proposition}

Another conjecture is that for every $d$-regular graph with no 4-cycles, $m(G,2) \le O(\frac{n}{d^2})$. For graphs of girth~5 Theorem~\ref{thm:girth} establishes a bound of $m(G,2) \le O(\frac{n}{d^{3/2}})$. We can improve over this bound as follows.

\begin{theorem}
\label{thm:4cycle}
Let $G$ be a graph of minimum degree $d$ and with no 4-cycles. Then $m(G,2) \le O(\frac{n}{d^{7/4}})$. Moreover, the contagious set can be chosen by the random-sequential algorithm.
\end{theorem}

For graphs of girth at least~7 (in fact, absence of 4-cycles and 6-cycles suffices), we can nearly obtain the desired upper bound of $O(\frac{n}{d^2})$, thus improving over the bounds implied by Theorem~\ref{thm:girth} for a wide range of girths. The algorithm used in the proof of Theorem~\ref{thm:girth7} involves an interplay between random and greedy selection of seeds.

\begin{theorem}
\label{thm:girth7}
Let $G$ be a $d$-regular graph of girth at least~7. Then $m(G,2) \le O(\frac{n\log d}{d^2})$.
\end{theorem}

One can combine a mild girth requirement with a modest expansion requirement and nearly obtain the desired upper bound of $O(\frac{n}{d^2})$. Observe that in Theorem~\ref{thm:girth_expander} we parameterize the spectral ratio $\lambda(G)/d$ by $1-\epsilon$. Hence for smaller $\epsilon$ we get worst expansion, and our upper bounds on $m(G,2)$ get \emph{larger}.

\begin{theorem}
\label{thm:girth_expander}
For arbitrary $\epsilon \in (0,1)$, let $G$ be an $(n,d,\lambda)$-graph with $\lambda \le (1 - \epsilon)d$ and with no 4-cycles. Then $m(G,2) \le O(\frac{n \log d}{\epsilon^2 d^2})$. Moreover, the contagious set can be chosen by a greedy algorithm. \end{theorem}

The proof of Theorem~\ref{thm:girth_expander} works without change when the condition $\lambda \le (1 - \epsilon)d$ is replaced by the weaker condition $\lambda_2 \le (1 - \epsilon)d$. Moreover, the contagious set in Theorem~\ref{thm:girth_expander} can also be chosen by the random-parallel algorithm, but the proof for this is more involved than the proof for the greedy algorithm, and is omitted.

Some of our upper bounds on $m(G,2)$ are summarized in Table~1. They hold for every graph with the corresponding expansion property.

\begin{table} [h!]
\begin{center}
    \begin{tabular}{ |l | c | c | }\hline
    \textbf{\small{{\bf Graph Parameters}}} & \small{{\bf Upper bound}} \\ \hline
    %Girth: $2k+1$        & \small{$O(nd^{\frac{1}{2^{k-1}}-2})$} \\ \hline
    Girth larger than $2\log\log d$        & \small{$O(\frac{n}{d^2})$} \\ \hline
    No 4-cycles    	                  & \small{$O(n d^{-7/4})$} \\ \hline
    Girth at least 7   	                  & \small{$O(\frac{n\log d}{d^2})$} \\ \hline
    %$\lambda(G):\lambda(G)<\delta d$, $\delta<1$ & \small{$O(nd^{-\frac{3}{2}})$} \\ \hline
    %$\lambda(G):\lambda(G)<\frac{1}{\sqrt{l}}d$ & \small{$O(\frac{n}{l^2})$} \\ \hline
    $\lambda(G) \le O(\sqrt{d})$ & \small{$O(\frac{n}{d^2})$} \\ \hline
    No 4-cycles and $\lambda(G) \leq (1 - \epsilon)d$ & \small{$O(\frac{\log d}{\epsilon^2 d^2}n)$} \\ \hline
    \end{tabular}\vspace{4pt}
    \caption{\textbf{Upper bounds on $m(G,2)$ as a function of graph parameters.} The results apply to $d$-regular graphs as a function of their girth and $\lambda(G)$, where $\lambda(G)$ is the second largest eigenvalue in absolute value. %The value of $l$ is limited to $l=O(d)$.
    }
    \label{table.runtimes} \vspace{-3pt}
\end{center}
\end{table}

One may ask what is the probable value of $m(G,2)$ for a random $d$-regular expander. For this purpose it is convenient to relax the regularity requirement, and analyze instead the standard binomial random graph model $G(n,p)$, in which each edge is present independently with probability $p$. For $d >> \log n$ and $p: = \frac{d}{n}$, these graphs are nearly $d$-regular (the degree of every vertex is roughly $d$), and furthermore, they are excellent expanders. For this distribution over nearly $d$-regular expanders we obtain a nearly tight characterization of the probable value of $m(G,2)$. Interestingly, it turns out that $m(G,2) \le o(\frac{n}{d^2})$.

%For random graphs with edge probability $p$ and $r=2$ we are able to improve upon the upper bound of $m(G,2) \le \frac{1+\delta}{2np^2}$, proven in \cite{JanLuc}, by a factor of $\log (np)$ and provide nearly matching lower bound:
\begin{theorem}
\label{thm:random}
Let $G\sim G(n,p)$ with $p: = \frac{d}{n}$ and $3 \log n < d < n^{\frac{1}{2}-\epsilon}$. Then with high probability
$$ \Omega\left(\frac{n}{d^2\log d}\right) \le m(G,2) \le O\left(\frac{n\log\log d}{d^2 \log d}\right).$$
\end{theorem}

Our current work is concerned with regular and nearly regular graphs. Dealing with highly irregular graphs is beyond the scope of the current paper. However, we remark here that the algorithmic question of finding a small contagious set in an irregular graph can be reduced to this question in regular graphs (though our reduction does not preserve expansion properties). See Section \ref{sec:hardness} for more details. We also note that insights from the study of contagious sets in expanding nearly regular graphs can be applied to expanding highly irregular graphs. See Section \ref{sec:irregular} for more details.
%In section \ref{sec:irregular} we explain why in general one cannot generalize in a straightforward fashion our results about random $d$-regular graphs to random graphs with average degree $d$ and also indicate the value of $m(G,r)$ for random graph models such as preferential attachment \cite{Barbasi}.
\subsection{Overview of proof techniques}

The following lemma simplifies the selection of contagious sets in spectral expanders (its proof is in Section~\ref{sec:spectral}). We remark that its proof works without change when the condition $\lambda \le \delta d$ is replaced by the weaker condition $\lambda_2 \le \delta d$.

\begin{lemma}
\label{lem:general_g}
Let $G$ be an $(n,d,\lambda)$-graph such that $\lambda<\delta d$ with $\delta<1$. Let the activation threshold of every vertex be $r=2$. Then every set of size larger than $\frac{n}{(1-\delta)d}$ is contagious.
\end{lemma}

Hence in spectral expanders it suffices to find a set that activates $\frac{n}{(1-\delta)d}$ vertices, and then the whole graph is activated by Lemma~\ref{lem:general_g}. A similar approach does not hold for graphs of large girth (which need not even be connected). For such graphs we shall use the random-sequential algorithm. We shall work in two stages, first finding a set of seeds that activates a large part of the graph, and then arguing that this suffices in order to activate the whole graph. However, now the second stage of the argument is more delicate and requires the selection of additional seeds.

\begin{lemma}
\label{lem:partialinfection}
Consider an arbitrary randomized algorithm $RA$ for selecting seeds in a graph $G$ with vertex set $[n]$. For every vertex $i$, let $p_i$ denote the probability that vertex $i$ is a seed, and let $q_i > 0$ denote the probability that vertex $i$ is activated. (Observe that necessarily $q_i \ge p_i$). Then there is a distribution $D$ over contagious sets such that for every vertex $i$, the probability that $i$ is a seed in a random contagious set selected according to $D$ is at most $p_i/q_i$.
\end{lemma}

\begin{proof}
Consider a sequence of rounds, where in every round $RA$ is applied on $G$ with independent randomness. As $q_i > 0$ for every $i$, eventually every vertex is activated in at least one of the rounds. For every $j$, include vertex $i$ in set $S_j$ if and only if $i$ was chosen as a seed in round $j$, and $i$ has not been activated in any round prior to $j$. The set $S = \bigcup S_j$ is necessarily contagious. (One can show by induction on $r$ that $\bigcup_{j=1}^r S_r$ activates all those vertices that are activated by round $r$.) Now:
$$Pr[i \in S] = \sum_{j = 1}^{\infty} Pr[i \in S_j] = \sum_{j=1}^{\infty} p_i(1 - q_i)^{j-1} = p_i \sum_{j=0}^{\infty} (1 - q_j)^j = p_i/q_i$$
\end{proof}

\begin{corollary}
\label{cor:partialinfection}
Let $G$ be a graph on $n$ vertices for which if every vertex is a seed independently with probability $p$, then for every vertex it holds that the probability that it is activated is at least~$1/C$ ($C>1$). Then $G$ has a contagious set of size at most $Cpn$.
\end{corollary}

\begin{proof}
Applying Lemma~\ref{lem:partialinfection} with $p_i = p$ and $q_i \ge 1/C$ we get for the random contagious set $S$:
$$E[|S|] = \sum_i Pr[i \in S] \le \sum_i C p = Cpn.$$
There must be at least one contagious set of size not larger than the expected size of contagious sets (taken from the distribution whose existence is implied by the proof).
\end{proof}
\medskip

We now explain how Theorem~\ref{thm:girth} (contagious sets in high girth graphs) is proved. As the girth of the graph is $2k+1$, every vertex $v$ is a root of a $d$-regular tree of depth $k$. Suppose that every leaf (a vertex at distance $k$ from $v$) is made a seed independently with probability $p$. Now we let a cascade of activations propagate from the leaves to the root, with the goal of inferring that the root is activated with constant probability. A simple calculation shows that once $p\ge \Omega(\frac{1}{d^2})$, we have ``amplification" in the sense that the probability of a node being activated increases as we get closer to the root of the tree. Hence, the deeper the tree, the smaller $p$ needs to be in order to ensure the root is activated with constant probability. Thereafter, an application of Corollary~\ref{cor:partialinfection} proves Theorem~\ref{thm:girth}.

Theorem~\ref{thm:Expander} (contagious sets in spectral expanders) follows from a proof similar to that of Theorem~\ref{thm:girth}, using a result of \cite{Csaba} that shows that every vertex of an $(n,d,\lambda)$-graph is a root of a sufficiently large tree (the minimal degree of a nonleaf node in the tree degree gets smaller when $\lambda$ approaches $d$ from below, hence our bounds deteriorate as $\lambda$ grows). The resulting algorithm is random-parallel rather than random-sequential because there is no need to use Corollary~\ref{cor:partialinfection} -- we can use Lemma~\ref{lem:general_g} instead. (Moreover, if one is not concerned with the number of generations until complete activation, it suffices to have the root of the tree activated with probability $\Omega(1/d)$ rather then constant, though this does not lead to substantial improvements in the bounds.)

The lower bound argument (Theorem~\ref{thm:LB}) is based on the observation that a ``small" contagious set $A$ entails a not much bigger set $B$ ($A \subset B$) such that $G[B]$ (the induced subgraph on $B$) has average degree close to~$4$. This is because every newly activated vertex in $B$ must be adjacent to two vertices causing it to become active. Hence it suffices to design $(n,d,\lambda)$-graphs with $\lambda = O(\sqrt{d})$ and large girth for which no set of $O(\frac{n}{d^{2 + \epsilon}})$ vertices has average degree (at least) nearly~$4$. Such graphs can be constructed using the probabilistic method.

The proof of Proposition~\ref{pro:general_bound} follows quite easily from Lemma~\ref{lem:general_g}.

The proof of Theorem~\ref{thm:4cycle} (contagious sets in graphs with no 4-cycles) is based on considering all neighbors of a vertex $v$ up to distance~3. However, as the girth is possibly smaller than~6, this neighborhood is no longer a tree, contrary to the case analyzed in Theorem~\ref{thm:girth}. Hence analyzing the probability that this neighborhood activates $v$ involves handling dependencies, making the analysis considerably more complicated than that of Theorem~\ref{thm:girth}. The absence of 4-cycles gives some control over these dependencies, leading to essentially the same amplification effect that one would get had the neighborhood been a tree.

The proof of Theorem~\ref{thm:girth7} (contagious sets in graphs with girth at least~7) involves selecting an initial set $A$ of $O(\frac{n\log d}{d^2})$ seeds, and considering the set $B$ of vertices that have a neighboring seed. Girth considerations are used in order to show that the subgraph induced on $B$ has large connected components. Thereafter, choosing one seed in each large connected component of $B$ activates the whole component. This allows us to cheaply extend the set of activated vertices to include most of $B$, and hence reach a size of $\Omega(\frac{n\log d}{d})$. At this stage one would expect a typical vertex to have $\Omega(\log d)$ active neighbors, and hence it should not be difficult to activate the remaining vertices in the graph. Turning this intuition into a formal proof involves some extra work, including appealing to Lemma~\ref{lem:partialinfection}.

The proof of Theorem~\ref{thm:girth_expander} (contagious sets in graphs with no 4-cycles and $\lambda = (1 - \epsilon)d$) involves the following amplification effect. Consider $\log d$ rounds, where in each round $n/d^2$ seeds are selected at random. The property that we wish to maintain is that the number of active vertices doubles after every round (until we eventually apply Lemma~\ref{lem:general_g}). Hence after every round $t$ we want there to be roughly $\frac{2^t n}{d^2}$ activated vertices (whereas there are only $\frac{t n}{d^2}$ seeds). For an inductive argument to apply, we would like the active vertices to have roughly $\frac{2^t n}{d}$ neighbors. These neighbors may be thought of as {\em excited} vertices, as they need only one additional active neighbor in order to become active. This makes it plausible that in the next round $\frac{2^t n}{d^2}$ new active vertices will be generated, because each new seed is  likely to have $2^t$ neighbors that are already excited, and these excited neighbors will be activated. We show that such a delicate balance can be kept for $\log d$ rounds by a greedy choice of seeds. Initially, our greedy rule does not seek to select a seed that maximizes the number of newly activated vertices, but rather to maximize the number of newly excited vertices. Both spectral expansion and absence of 4-cycles are used in order to analyze this greedy rule. Only after the number of excited vertices reaches $n/2$, we switch to a greedy rule that maximizes the number of newly activated vertices.

The proof of the upper bound in Theorem \ref{thm:random} is based on selecting an initial small set $A$ of seeds, and then considering its external neighborhood $\partial(A)$, namely, the set of {\em excited} vertices (also considered in Theorem~\ref{thm:girth_expander}). Given that the graph is random, one can analyze the distribution of the sizes of the connected components of the subgraph induced by $\partial(A)$. % and conclude that a reasonable fraction of the exited vertices are in connected components of size $k$, where $k = \Omega(\frac{\log d}{\log\log d})$.
Introducing a single seed in a connected component of size $k$ then activates the whole component, thus giving $k$ activated vertices per investment of one seed.
%thus is based on the observation that given a set $A$ in $G:=G(n,p)$ one can derive lower bound on the number of vertices contained in connected components of size $k$ in $N(A)$-the external neighborhood of $A$. Observe once $A$ is activated, activating a single vertex in a connected component $C$ lying in $N(A)$ will activate the whole of $C$.
By choosing the parameters $|A|$ and $k$ appropriately it turns out that we can activate a set of size $\frac{n}{d^2} = \frac{1}{np^2}$ in $G$ by choosing $O\left(n\frac{\log\log (d)}{d^2 \log d}\right)$ seeds. Thereafter, the results of~\cite{JanLuc} can be used in order to deduce that $G$ is activated with high probability.

The lower bound in Theorem~\ref{thm:LB} and the lower bound in Theorem~\ref{thm:random} both apply to random graphs, but the graphs in Theorem~\ref{thm:LB} are required to be regular whereas those in Theorem~\ref{thm:random} are only nearly regular. The difference in the random graph model makes the analysis in Theorem~\ref{thm:random} easier, leading to a higher lower bound. Both lower bounds involve a union bound over an exponential number of potential ``bad events", but in the proof of Theorem~\ref{thm:random} we can use an approach from~\cite{JanLuc} that allows us to reduce the number of bad events considered, hence leading to better bounds.
%As it turns out, for a set of size $O(\frac{1}{np^2\log (np)})$ to be contagious a sum of independent random variables needs to exceed the expectation of the sum by a significant factor. Using concentration results along with a union bound argument shows that the existence of such a small contagious set occurs with probability $o(1)$.

\subsection{Related work}

As already noted, $m(G,r)$ has been determined for certain families of graphs. For example, if $G$ is the $k$-dimensional grid $[n]^k$ then $m(G,r)=\Theta(n^{r-1})$ if  $1 \le r \le k$ and $\Theta(n^k)$ otherwise \cite{square}. If $G$ is the $n$-dimensional hypercube on $2^n$ vertices it is known that $m(G,2)=n$ \cite{BB}.
To the best of our knowledge, the current work is the first to study how $m(G,r)$ depends on the girth of $G$ and on $\lambda(G)$.

Random regular graphs are expected to have very good expansion properties, and hence results on $m(G,r)$ for random regular graphs can serve as a benchmark against which to compare results for expanders.
Balogh and Pittel~\cite{BP} proved an upper bound on $m(G,r)$ when $G$ is chosen uniformly among all $n$-vertex $d$-regular graphs. Using differential equations, they show that a random set of size smaller than $(p(G,r)-\epsilon_n)n$ will not be contagious with high probability. On the other hand, a random set of size $(p(G,r)+\epsilon_n)n$ will be contagious with high probability, where $\lim_{n \rightarrow \infty}\epsilon_n=0$ (for some explicitly defined function $\epsilon_n$). The value of $p(G,r)$ is $1-\inf_{y \in (0,1)}\frac{y}{R(y)}$ with $R(y)=\Pr({\rm Bin} (d-1,1-y)<r)$ where ${\rm Bin}(d-1,1-y)$ is a binomial random variable with parameters $d-1$ and $1-y$. It can be shown that $p(G,2)$ tends to $\frac{1}{2d^2}$ as $d$ grows~\cite{BP}. We are not aware of a closed formula of $p(G,r)$, nor are we aware of asymptotic evaluations (as a function of $d$ and $r$) of it for $2<r<d-1$. The work of~\cite{BP} on random $d$ regular graph does not provide lower bounds on $m(G,r)$ --  it only implies that with high probability (probability $1-o(1)$) a \emph{random} set of size $(\frac{1}{2d^2}-\epsilon)n$ is not contagious. %these results do not establish that w.h.p. \emph{no} set of size $(\frac{1}{2d^2}-\epsilon)n$ is contagious.

A different proof of the result of \cite{BP} building on cores in random graphs was given by Janson \cite{Janson}. Interestingly, $p(G,r)$ is identical to the critical threshold for complete activation of the infinite $d$-regular tree~\cite{peres}. Our bounds for expander graphs are partly based on analyzing the spread of activation from the leafs of a $d$-regular tree to its root, and this part of the analysis involves a recursive approach similar to those employed in previous work (though we do so in a setting in which the depth of the tree is finite rather than infinite).

The critical size of a random set required for complete activation of $G(n,p)$ for arbitrary constant threshold $r$ and $\frac{1}{n}\ll p \ll \frac{1}{\sqrt{n}}$ (for this range of parameters the resulting graph is likely to be nearly regular) was determined in great detail of precision by \cite{JanLuc}. In particular, it is shown that for $\frac{3\log n}{n}<p\ll\frac{1}{\sqrt{n}}$, activating a set of cardinality at least $\frac{1+\delta}{2np^2}$ vertices where $\delta>0$ is a fixed constant activates the entire graph with high probability. In particular this work implies that when $G\sim G(n,p)$ with $p$ as above, $m(G,2)\leq \frac{1+\delta}{2np^2}$.
%To the best of our knowldege this is the best upper bound on $m(G,2)$ in random graphs.
It was also shown that a random set of size $\frac{1-\delta}{2np^2}$ is unlikely to activate $G(n,p)$. Our Theorem~\ref{thm:random} (whose proof involves a more sophisticated choice of set of seeds) implies that for such graphs $m(G,2)$ is almost surely significantly smaller than the bounds implied by~\cite{JanLuc}, and also provides the first lower bound on $m(G,2)$ in such graphs.

The time (number of generations) until complete activation in bootstrap percolation is the topic of several recent works such as ~\cite{Time}. For $G(n,p)$, Janson et al., \cite{JanLuc} studied the number of generations until complete activation for various parameters (e.g., Theorem 3.10, pp. 2000). In particular, for $r=2$, they show that when $p=n^{-\alpha}$ where $1/2<\alpha<1$ and for a fixed set of size $\frac{1+\delta}{np^2}$ (namely., a set of cardinality twice as large than the critical cardinality needed for complete activation), the number of generations is with high probability $\log \log (np)+O(1)$.

The optimization problem, where given $G=(V,E)$ with threshold $r$, we seek to activate a set of minimum cardinality (that is, of cardinality $m(G,r)$) so that the whole of $G$ is activated, is called the \emph{Target Set Selection} problem \cite{Chen09}.  Calculating $m(G,r)$ exactly is NP-hard and obtaining an approximation better than $O(2^{\log^{1-\epsilon}n})$ ($n$ is the number of vertices) is intractable, unless $NP \subseteq DTIME(n^{poly(\log n)})$ \cite{Chen09}. These hardness results hold even when $r=2$ and $G$ has maximal degree $d$, where $d$ is a constant not depending on the size of $G$ \cite{Chen09}. For recent results demonstrating the tractability of target set selection in graphs with certain structural properties such as bounded treewidth see \cite{Ben_Zwi,Chopin}.  To the best of our knowledge, no approximation algorithm with approximation ratio significantly better than the trivial $n$ approximation is known for the target set selection problem. The results of \cite{Ackerman,Reichman} are algorithmic and they imply for a fixed threshold $r$ a polynomial time $O(n/d)$ approximation algorithm for $m(G,r)$. We are not aware of an approximation algorithm  achieving better approximation ratio as a function of $d$ for $m(G,r)$ in $d$-regular graphs. Approximation and hardness of other propagation problems that are similar to target set selection was considered in \cite{Aazami}.
\section{Preliminaries and notation}

Unless explicitly stated, we will always deal with $d$-regular, undirected graphs on $n$ vertices. The reader may think of $d$ as a large constant independent of $n$, though all results easily extend to the case that $d$ is some  growing function of $n$ (with some self-evident upper bounds on the rate of growth of $d$ as a function of $n$, that depend on the nature of the result). A graph $G$ has \emph{girth} $g$ if the shortest cycle in $G$ is of length $g$. For clarity reasons, floor and ceiling signs are omitted. For a natural number $l$, we denote the set $\{1,...,l\}$ by $[l]$. $\log$ refers to the logarithm in base 2.

Given a $d$-regular graph $G=(V,E)$ in the bootstrap percolation model with threshold $r$, we shall often be interested in the case where every vertex is chosen to belong to $A_0$ independently with probability $p_0 \in [0,1]$. We denote by $p_c(G,r)$ the minimal $p_0$ such that a set $A_0$ whose elements are chosen independently with probability $p_0$ is contagious with probability $\frac{1}{2}$.
$$p_c(G,r)=\inf_p [\Pr(\langle A_0 \rangle =V)=\frac{1}{2}],$$
where every vertex is chosen independently to $A_0$ with probability $p$.
Observe that we always have that $m(G,r) \leq p_C(G,r)\cdot n$. In general $m(G,r)$ may be much smaller than $p_c(G,r)\cdot n$. For example, for the hypercube over $2^n$ vertices, $m(G,2)=n$ whereas $p_c(G,2)=\Theta(\frac{2^{-2\sqrt{n}}}{n^2})$ \cite{BB}.

Given a vertex $v$ and a set $S$, the number of neighbors of $v$ in $S$ is denoted by $deg_S(v)$.
For two sets of vertices $A$ and $B$ let $e(A,B)$ be the number of ordered pairs of vertices $(u,v)$ with $u \in A$, $v \in B$ and $(u,v)$ in $E$ ($A,B$, need not be disjoint). We denote by $e(A)$ the set of all edges whose two endpoints belong to $A$. For a subset $A$ of vertices, we denote by $\partial(A)$ the set of all vertices in $V \setminus A$ having a neighbor in $A$ and by $N(A)$ the set of all vertices in $V$ having a neighbor in $A$.  The adjacency matrix of an $n$--vertex graph $G$, $A_G$, is symmetric hence it has $n$ real eigenvalues. Let $\lambda_1\geq\lambda_2\geq\ldots\geq\lambda_n $ be the eigenvalues of $A_G$. It is known that $\lambda_1=d$ and for every $i>1$, $|\lambda_i| \leq d$ (see for example \cite{KS}). Let $\lambda(G)=\max \{|\lambda_2|,|\lambda_n|\}$. We say that $G$ is an $(n,d,\lambda)$-graph if $G$ is $d$-regular and $\lambda(G) \le \lambda$. We will focus on the case that $\lambda$ is smaller than $\delta d$ where $\delta<1$.
%For concreteness we often choose $\delta<0.01$, though we will sometimes give explicit bounds that ensure the property we need.

The following Lemma relates edge expansion to $\lambda_2$, the second largest positive eigenvalue of $G$. The proof can be found in \cite{AS}.
%$\lambda(G)$ is related to the expansion of $G$ as follows:
\begin{lemma}
\label{lem:expansion}
Let $G=(V,E)$ be a $d$-regular graph. Then, for every partition $B$,$C$ of $V$, $e(B,C) \geq \frac{(d-\lambda_2)|B||C|}{n}$.
\end{lemma}

A graph is called an \emph{expander graph} if for every set of vertices $W$ of size at most $n/2$, the set $\partial(W)$ is of size at least $c|W|$ with $c>0$ independent of $n$. It can be verified that if $G$ is a $(n,d,\lambda)$ graph with $\lambda \le \delta d$ ($\delta<1$) then $G$ is an expander with $c$ being at least $\frac{d-\lambda}{2d}$ (see \cite{AS}, Corollary 9.2.2).

We shall use Azuma's inequality to prove concentration results.

\begin{lemma} \label{lem:Azuma}
Let $X_0,...,X_n$ be a martingale such that for every $1 \leq k <n$ it holds that $|X_k-X_{k-1}|\leq c_k$. Then for every nonnegative integer $t$ and real $B >0$
$$\Pr(|X_t-X_0|\geq B) \leq 2\exp\left(\frac{-B^2}{\sum_{i=1}^tc_i^2}\right).$$
\end{lemma}

We shall sometimes use the term {\em infected} to describe an activated vertex that is not one of the {\em seeds}, but rather became activated by having at least $r$ active neighbors. When $r=2$, the term {\em excited} describes a non-active vertex that has $r-1$ active neighbors.

\section{Contagious sets in graphs with large girth}

In this section we focus on the case where the threshold $r$ of every vertex equals $2$. We derive upper bounds on $m(G,2)$ as a function of the girth of $G$. We do this by using bounds on bootstrap percolation on $d$-regular trees. It is known and easy to see that if one considers an infinite rooted tree in which every vertex has $d$ children, the following holds. Let $p_0 = p$ denote the initial activation probability, let $p_i$ denote the probability that the root becomes activated by generation at most $i$ of the bootstrap percolation process, and let $q_i = 1 - p_i$. Then for $i \ge 1$, $q_i = q_0((q_{i-1})^d + dp_{i-1}(q_{i-1})^{d-1})$. Using this recursive relation it is not difficult to show that for $p = c/d^2$ (for a sufficiently large value of $c$) we have $p_k = \Omega(1)$ already for some $k = \log\log d + O(1)$, and $p_k = 1 - o(1/n^2)$ already for some $k = O(\log_d\log n) + \log\log d + O(1)$. The following lemma provides a short proof of these statements in which no attempt was made to optimize the constants involved. For simplicity, given a finite tree, the lemma only uses the assumption that the leaves are initially activated with probability $p$, ignoring the fact that also internal vertices may be initially activated.

\begin{lemma}\label{lem:tree}
Let $T_{d,k}$ be the complete $d$-regular tree (e.g., the root being of degree $d$ and all other nonleaf vertices are of degree $d+1$) of depth $k$, with $d$ being sufficiently large.

\begin{enumerate}
\item Suppose every leaf of the tree is activated independently with probability $p = \frac{g(k)}{d^2}$ with   $g(k)=10 d^{\frac{1}{2^{k-1}}}$. Then the probability the root is activated once we apply the bootstrap percolation process is at least $\frac{1}{2e}$. As a special case, if $k > \log\log d + 1$ then a value of $p = O(\frac{1}{d^2})$ suffices in order to activate the root with probability at least $\frac{1}{2e}$.

\item If $k=C\log_d \log n+ \log \log d + O(1)$ (for a sufficiently large absolute constant $C$) then a value of $p = O(\frac{1}{d^2})$ suffices in order to activate the root with probability at least $\frac{1}{n^2}$.
\end{enumerate}
\end{lemma}
\begin{proof}
%Set $p=10p'$ (recall $p$ is the probability of a leaf becoming active).
A vertex in $T_{d,k}$ is said to be in level $\ell$ with $0 \le \ell \le k$ if its distance from the root is $\ell$. Hence the root is in level~0 whereas the leaves are in level~$k$.
Let $p_i$ ($0 \leq i \leq k$) be the probability that a vertex in level $k-i$ gets activated.
Hence $p_0=p$ and $p_k$ is the probability of the root being activated in the bootstrap percolation process. We shall write $p_i = \frac{g_i}{d^2}$ with $g_0 = g(k)$ as defined in the lemma.  An internal vertex $w$ of the tree becomes activated if it has at least two active children. Hence for $j<k$,
$p_{j+1} \geq \Pr({\rm Bin}(d,p_j) \geq 2) \geq {d \choose 2}{p_j}^2(1-p_j)^{d-2}$, with ${\rm Bin}(d,p_j)$ the binomial distribution with parameters $d$ and $p_j$. Hence $g_{j+1} \ge \frac{1}{2}(g_j)^2 (1-p_j)^{d-2}$.
As long as $p_j \leq \frac{1}{d}$ then we have that $g_{j+1} \ge \frac{1}{3}(g_j)^2\frac{1}{e} \ge \frac{1}{10}(g_j)^2$, and by induction we have that
$$p_{i} \geq 10(\frac{g_0}{10})^{2^i}d^{-2} = 10d^{\frac{1}{2^{k-1-i}}}d^{-2}.$$

Substituting $i = k-1$, children of the root have probability at least $\frac{1}{d}$ of being activated, implying that $p_k \ge \frac{1}{2e}$, proving item~1 of the lemma.

We now prove item~2 of the lemma. By item~1, every vertex $v \in T$ in level $k - \log \log d - 1$ gets activated with probability at least $\frac{1}{2e}$.   We now use the inequality $p_{j+1}>1-(1-p_j)^d-p_jd(1-p_j)^{d-1})$ that holds for every $j$.
%where $(C/2 \log \log d \le j \le C \log \log n)$. We get that $p_j>1-(1-p)^d-pd(1-p)^{d-1}).$
Let $q_j=1-p_j$. Then $q_{j+1}\leq q_j^{d-1}(d+1) \leq q_j^{d/2}$, where the last inequality applies in the range that $j > \log \log d + 1$ and $d$ is sufficiently large. We get by induction that $q_i \leq e^{-(\frac{d}{2})^i}$. Now we consider two cases. If $d < \log n$, then when $i= \log\log d + C \log_d \log n$ and $C$ is sufficiently large the probability the root is \emph{not} infected is at most $\frac{1}{n^2}$. If $d \ge \log n$, the same consequence is obtained by taking $i = \log\log d + O(1)$. In either case, a union bound over the $n$ vertices of the graph implies item~2 of the lemma.
%If $d\geq \log n$, then every vertex in level $k-\log\log d$ remains inactive with probability at most $\log^2 n(\frac{1}{2e})^2(1-\frac{1}{2e})^{d-2}\leq n^{-\beta}$ for some $\beta \in (0,1)$. Then an arbitrary vertex in level $k-\log \log d+1$ is active with probability at least $1-((n^{-\beta})^d+d(n^{-\beta})^{d-1})=1-o(n^{-2})$, as desired. Hence we can assume $d<\log n$.
\end{proof}
\medskip

We can now present a proof of Theorem~\ref{thm:girth}:
\begin{proof}
Observe that as the girth of $G$ is $2k+1$, every vertex is the root of a $(d-1)$-regular tree of depth $k$. The assertion in the theorem now follows from Lemma \ref{lem:tree} and Corollary \ref{cor:partialinfection}.
\end{proof}

We remark that when $G$ is a $d$-regular graph of order $n$ and girth $\Omega(\log \log n)$ then Lemma \ref{lem:tree} implies that $p_c(G,2)=O(\frac{1}{d^2})$.

\section{Bounds for $m(G,2)$ in spectral expanders}
\label{sec:spectral}

In this section we concentrate on $(n,d,\lambda)$-graphs. Our main goal is to derive upper bounds on $m(G,2)$ in terms of $\lambda(G)$. We start by proving Lemma~\ref{lem:general_g}.

%\begin{lemma}
%\label{lem:general_g}
%Let $G$ be an $(n,d,\lambda)$-graph such that $\lambda<\delta d$ with $d$ large enough and $\delta<1$. Let the activation threshold of every vertex be $r=2$. Then every set of size larger than $\frac{n}{(1-\delta)d}$ is %contagious.
%\end{lemma}

\begin{proof}
Consider a set $S$ of size $|S|$ that is not contagious. We can assume without loss of generality that $S$ is inclusion-maximal with respect to being active (namely, \emph{every} vertex \emph{not} belonging to $S$ is not active). For every $u \in V \setminus S$ it holds that $deg_S(u) \leq 1$. Thus $e(S, V\setminus S) \leq |V \setminus S|=n-|S|$. On the other hand, by Lemma ~\ref{lem:expansion} $$e(S, V\setminus S)  \geq \frac{(1-\delta)d|S|(n-|S|)}{n}.$$
Combining these inequalities we have that
$$\frac{(1-\delta)d|S|(n-|S|)}{n} \leq n-|S|.$$
Hence $|S| \leq \frac{n}{(1-\delta)d}$. As required.
\end{proof}
\medskip
%Note: the lemma above can be strengthened at the expense of strengthening the condition on $\delta$. See Lemma ~\ref{lem:improvedsecondstep}.

Using Lemma \ref{lem:general_g} we first prove Proposition~\ref{pro:general_bound}.

\begin{proof}
Activate independently every vertex with probability $p$ (where $p$ will be chosen later). Let $A_1$ denote the set of non-seed vertices that have at least two seed neighbors (and hence become active), and let $p_1$ denote the probability that a vertex belongs to $A_1$. Then
$$p_1 \ge (1 - p){d \choose 2}p^2(1-p)^{d-2}.$$ Assuming $d$ is sufficiently large and $p$ is smaller than $\frac{1}{d}$ we get that
$p_1 \geq \frac{(dp)^2}{4}$. By Lemma \ref{lem:general_g}, every set of size $\frac{cn}{d}$, where $c>\frac{1}{1-\delta}$, is contagious. If $p > \frac{4\sqrt{c}}{d^{\frac{3}{2}}}$ we get that the \emph{expected} number of vertices in $A_1$ is at least $\frac{2cn}{d}$. We proceed and show that w.h.p. $|A_1| > \frac{cn}{d}$ vertices.
%By the fact that the number of active vertices is at most $n$ it follows that with probability at least $\frac{c}{d}$ there indeed are at least $\frac{cn}{d}$ active vertices. Note that the number of seeds used is concentrated around its expectation
Define the familiar Doob exposure martingale, e.g., exposing the set of seeds according to some predetermined order and considering the expected number of vertices in $A_1$. Observe that whether an exposed vertex is a seed or not can effect at most $d$ neighboring vertices. We get using Lemma~\ref{lem:Azuma} (Azuma's inequality) that for such $p$ with high probability $|A_1| \ge \frac{cn}{d}$. The Lemma follows.
\end{proof}
\medskip

We now turn to prove Theorem~\ref{thm:Expander}.
The proof of Theorem \ref{thm:girth} can be generalized to the case where every vertex is contained in a regular tree of degree $\Omega(l)$ and sufficiently large depth (even if the tree is not induced). %Spectral expanders are known to contain large trees.
There is a long line of research concerned with embedding trees in expanders, starting with the works of P\'{o}sa \cite{posa} and Friedman and Pippenger \cite{Friedman}. We will use the recent result of Balogh, Csaba, Pei and Samotij \cite{Csaba}, building on the work of Haxell \cite{Haxell}.

\begin{theorem} [Theorem 5 in~\cite{Csaba}] \label{thm:tree_embed}
Let $l \geq 2$ and $\epsilon \in (0,\frac{1}{2})$. If $\lambda < \frac{\epsilon d}{\sqrt{8l}}$ then every $(n,d,\lambda)$-graph contains every tree of order at most $(1-\epsilon)n$ and maximum degree $l$. Furthermore, for every vertex $v \in G$, fixing a (rooted) tree $T$ satisfying these conditions, $T$ can be embedded into $G$ with $v$ being the root of $T$.
\end{theorem}

We now prove Theorem~\ref{thm:Expander}.

\begin{proof}
By Theorem~\ref{thm:tree_embed} every vertex is the root of a regular tree of degree $\Omega(l)$ of depth $k=\Omega(\log_l \log n+\log \log d)$. The proof of Lemma~\ref{lem:tree} then implies that if every vertex in $G$ is activated independently with probability $p \geq \Omega(\frac{1}{l^2})$, then for every vertex $v$ in $G$ the probability $v$ is not activated by the bootstrap percolation process is $O(\frac{1}{n^2})$. Hence the entire graph is activated with high probability by taking union bounds over all vertices. Furthermore, it is immediate that the number of generations until complete activation is $O(\log_l \log n+\log\log d)$.
\end{proof}
\medskip
The proof of Proposition~\ref{pro:generations} is based on elementary probabilistic arguments.

\begin{proof}
Consider an arbitrary $d$-regular graph. For a fixed vertex $v$ there are at most $d(d-1)^{\log_d \log n - 1} < \log n$ vertices of distance $\log_d \log n$ from $v$. Vertex $v$ is activated within $\log_d \log n$ generations only if at least one vertex (possibly $v$ itself) within its $\log_d \log n$ neighborhood is initially activated. A simple greedy argument shows that there is a set $U$ of at least $n/(\log n)^2$ vertices in $G$ such that the distance between any two vertices of $U$ is at least $2\log_d \log n$. Hence for every two vertices in $U$, the events that they are activated within $\log_d \log n$ generations are independent. It follows that if every vertex is initially activated independently with probability $1/4$, the probability that all vertices of $U$ are activated in $\log_d \log n$ generations is at most:
$$\left(1 - (\frac{3}{4})^{\log n}\right)^{n/(\log n)^2} = o(1).$$
\end{proof}
\medskip
%We remark that it is not hard to apply the results of \cite{Friedman} to deduce similar results to Theorem ~\ref{thm:large_girth} to $d$-regular graphs with the property  that for every subset of vertices $S$ of size at most $O(d^{\lg \lg d}),$ satisfies $|\partial(S)| \geq l$.

We now turn to Theorem~\ref{thm:LB}, exhibiting $d$-regular expanders for which $m(G,2)=\Omega(\frac{n}{d^{2+\epsilon}})$. Our lower bound on $m(G,2)$ is based on the following lemma.

\begin{lemma}
\label{lem:subgraph}
Suppose an $n$-vertex graph $G=(V,E)$ has a contagious set of size $t_0$. Then for every $t$ such that $t_0 \leq t \leq n$ there is a subgraph of $G$ induced by $t$ vertices, spanning at least $2(t-t_0)$ edges.
\end{lemma}

\begin{proof}
Let $A_0$ be a contagious set of size $t_0$. Then there exists an ordering of the vertices of $V \setminus A_0$, $v_1,...,v_{n-t}$ such that $\forall i, 1 \leq i \leq n-t_0$, $v_i$ is connected to at least two vertices in $A_0 \cup \{v_1,...,v_{i-1}\}$. Given $t_0 \leq t$, let $B_t$ be $A_0 \cup \{v_1,...,v_{t-t_0}\}$. Then $2(t-t_0) \leq |E(B_t)|$. As required.
\end{proof}
\medskip

Lemma~\ref{lem:subgraph} implies that in order to prove lower bounds on $m(G,2)$ it suffices to exhibit graphs that do not have small subgraphs of average degree nearly~4. To exhibit expander graphs that do not have small subgraphs of average degree nearly~4 we apply the probabilistic method.
For the expansion property, we shall use the following theorem of Friedman~\cite{friedman03}.

\begin{theorem}[Friedman~\cite{friedman03}]
\label{thm:friedman}
For arbitrary $\delta > 0$, a random $d$-regular graph $G$ has probability $1 - o(1)$ (the $o(1)$ term tends to~0 as $n$ grows) of satisfying $\lambda(G) \leq 2\sqrt{d-1}+\delta$.
\end{theorem}

We remark that the bound in Theorem~\ref{thm:friedman} matches (up to low order terms) the lower bound on $\lambda$ for arbitrary $d$-regular graphs.

We now find it convenient to temporarily switch to the configuration model $G^*(n,d)$ of random $d$-regular multigraphs (see for example \cite{Wormald}). Let $nd$ be even, the vertex set of the sampled graph be $[n]$, and let $d$ be a constant independent of $n$. Let $W=[n]\times [d]$. Elements of $W$ are called cells. For $i \in [n]$ we define $W_i$, as the set $\{i\}\times [d]$. Now we generate $G$ by choosing a uniform perfect matching over all matchings of all cells in $W$. Suppose a cell from $W_i$ is matched to a cell in $W_j$: in this case we add an edge between two vertices $i,j \in [n]$. Observe that the resulting graph need not be simple and may contain multiple edges and self loops. However, we shall use the following known theorem (see for example \cite{Wormald}).

\begin{theorem}
\label{thm:wormald}
A graph $G$ sampled from $G^*(n,d)$ is simple (has no parallel edges and no self loops) with probability tending to $e^{-(d^2-1)/4}$ (which is bounded away from $0$ for a constant $d$) as $n$ tends to infinity. Conditioned on being simple, $G$ is distributed as $G(n,d)$. Namely, $G$ is a uniform sample of a $d$-regular $n$-vertex graph.
\end{theorem}

As edges in $G^*(n,d)$ are not independent, we shall use the following known lemma:

\begin{lemma}\label{lem:config}
Let $G=(V,E)$ be a graph sampled from $G^*(n,d)$. Let $E_0$ be a set of $k$ distinct unordered pairs $e_1,..,e_k$ where each pair consists of two distinct vertices in $V$ where $k<\frac{nd-1}{4}$. Then the probability that $e_1,...,e_k$ simultaneously belong to $E$ is bounded by $(\frac{2d}{n})^k$.
\end{lemma}
\begin{proof}
In the configuration model, fix $\widetilde{e}_1,..,\widetilde{e}_k$ with $\widetilde{e}_i$ being an edge connecting a \emph{fixed} cell in $W_r$ to a \emph{fixed} cell in $W_s$ where it is assumed that $e_i$ is between the vertices $r$ and $s$ ($r,s \in V$). Then the probability that $\widetilde{e}_1,..,\widetilde{e}_k$ all exist in the configuration model is exactly $\frac{1}{nd-1}\cdot\frac{1}{nd-3}\cdot...\cdot\frac{1}{nd-2k+1}$ which is bounded by $(\frac{2}{nd})^k$. The lemma follows as for each $i \leq k$, conditioned on $e_1,..,e_{i-1}$ chosen there are at most $d^2$ choices for cells realizing $e_{i}$.
\end{proof}
\medskip

\begin{theorem}
\label{thm:Deadline}
For every $\epsilon>0$ there exists $d_0$ such that for every $d>d_0$ if $G$ is sampled from the configuration model $G^*(n,d)$, then $m(G,2) \ge \frac{n}{d^{2+\epsilon}}$ with probability $1 - o(1)$, where the $o(1)$ term tends to~0 as $n$ grows.
\end{theorem}
\begin{proof}
Set $t_0=\frac{n}{d^{2+\epsilon}}$, $t=C\cdot t_0$, where $C=C(\epsilon)>2$ will be determined later.
The probability $G$ sampled from $G^*(n,d)$ contains a subgraph of size $t$ spanning at least $2(t-t_0)$ edges is at most
$${n \choose t}{{t \choose 2} \choose 2(t-t_0)}\left(\frac{2d}{n}\right)^{2(t-t_0)}$$
$$\leq \left(\frac{en}{t}\right)^t\left(\frac{(Ct_0)^2}{Ct_0}\right)^{2(C-1)t_0} \left(\frac{2d}{n}\right)^{2(t-t_0)}$$
$$\leq \left[\left(\frac{en}{Ct_0}\right)^C\left(\frac{Ct_0 2d}{n}\right)^{2(C-1)}\right]^{t_0}$$
$$=\left[O(1)\frac{t_0^{C-2}d^{2C-2}}{n^{C-2}}\right]^{t_0},$$
where we consider terms depending only on $C$ as $O(1)$ since we are taking $d$ to be large enough.
Substituting $t_0=\frac{n}{d^{2+\epsilon}}$ the above expression simplifies to
$$\left[O(1)d^{2c-2-(2+\epsilon)(c-2)}\right]^{t_0}.$$ Taking $C(\epsilon)=\frac{3}{\epsilon}$ (we can assume $\epsilon$ is sufficiently small) we get that the probability that $G$ contains a subgraph of size $t$ spanning at least $2(t-t_0)$ edges is $o(1)$, for $t_0=\frac{n}{d^{2+\epsilon}}$. Lemma~\ref{lem:subgraph} then implies that the probability that $G$ has
a contagious set of cardinality smaller than $\frac{n}{d^{2+\epsilon}}$ is $o(1)$.
\end{proof}
\medskip

We can proceed and prove Theorem \ref{thm:LB}:
\begin{proof}
Sample at random a graph $G$ from $G^*(n,d)$. By Theorem~\ref{thm:Deadline} we have that $m(G,2) \ge \frac{n}{d^{2+\epsilon}}$ with probability $1 - o(1)$. By Theorem~\ref{thm:wormald}, $G$ is simple with probability bounded away from~0. Hence conditioned on $G$ being simple, the probability that it fails to have $m(G,2) \ge \frac{n}{d^{2+\epsilon}}$ is still $o(1)$. Conditioned on being simple, Theorem~\ref{thm:friedman} implies that $G$ fails to have $\lambda(G)=O(\sqrt{d})$ with probability $o(1)$.  For a fixed integer $k$ it is known, that with probability $p(d,k)>0$ (where $p(d,k)$ depends only on $d,k$ but not on $n$) a random $d$-regular graph has girth at least $k$ (see for example, \cite{Wormald}). Hence there is positive probability that $G$ is simultaneously simple, of girth at least $k$, has $\lambda(G)=O(\sqrt{d})$, and moreover,  $m(G,2) \ge \frac{n}{d^{2+\epsilon}}$. This proves Theorem \ref{thm:LB}.
\end{proof}
\medskip
%Another question is how large $m(G,2)$ can be in $d$-regular graphs in the absence of small cycles. It is well known (e.g., \cite{Wormald}) that for any fixed $k$, the number of cycles of length $k$ in $G(n,d)$ is bounded with high probability by a constant $C(k)$ not depending on $n$. Using this fact we can prove lower bounds on $m(G,2)$ for high girth graphs:

%\begin{theorem}
%\label{thm:LBgirth}
%Let $\epsilon \in (0,1)$ be sufficiently small, $d$-large enough, and $k$ a fixed constant. Then there exist $d$-regular graphs that are have girth at least $k$ such that $m(G,2)=\Omega(\frac{n}{d^{2+\epsilon}})$.
%\end{theorem}
%\begin{proof}
%Consider a random $d$-regular graph $G$. It is known that with probability bounded away from $0$, $G$ has girth $k$ . As the event $m(G,2)=\Omega(\frac{n}{d^{2+\epsilon}})$ holds with high probability when $G$ is a random regular graph, it follows that with positive probability $G$ has girth $k$ \emph{and} $m(G,2)=\Omega(\frac{n}{d^{2+\epsilon}})$. AS required.
%\end{proof}

\section{Contagious sets in graphs with no 4-cycles}
\label{sec:4cycle}

We have seen that for $d$-regular graphs $m(G,2)$ may be at least $\frac{2n}{d+1}$. It is not hard to construct triangle free graphs with $m(G,2)$ at least $\frac{n}{d}$ (take $\frac{n}{2d}$ disjoint copies of complete bipartite $d$-regular graphs). In this section we show that situation is different for graphs without 4-cycles, proving Theorem~\ref{thm:4cycle}.

Given a graph $G$ of minimum degree at least $d+1$ (for notational reasons, we find it easier in this section to work with degree $d+1$ as opposed to $d$), a vertex $v$ and and a parameter $k \ge 0$, a $(d,k)$-tree rooted at $v$ is a $d$-ary tree of depth $k$ that can be defined by induction on $k$ as follows. A $(d,0)$-tree is $v$ itself. A $(d,1)$ tree has $v$ as its root, and $d$ distinct neighbors of $v$ as its leaves. Thereafter, a $(d,k+1)$-tree is obtained from a $(d,k)$-tree as follows: every leaf of the $(d,k)$-tree gets $d$ of its neighbors in $G$ (excluding its parent node in the tree) as children in the $(d,k+1)$-tree. Hence for every node in a $(d,k)$-tree, all its tree neighbors are distinct vertices of $G$. However, the same node of $G$ may appear multiple times in the $(d,k)$-tree.

For a vertex $v$ and $k \ge 0$, a $k$-witness is a $(2,k)$-tree rooted at $v$ in which all its leaves are seeds. A $k$-witness implies that $v$ is activated, by propagating activations from the leaves to the root. Observe that we do not require the leaves to represent distinct vertices of $G$, or to represent vertices different from internal nodes of the tree. Observe also that $v$ might be activated without there being any $k$-witness to its activation (for example, by having one neighbor of $v$ as a seed and another neighbor of $v$ activated by two of its seed neighbors).

\begin{proposition}
\label{pro:2ind}
Consider a $(d,k)$-tree $T$ rooted at $v$. Then the number of $(2,k)$-trees rooted at $v$ that $T$ contains is ${d \choose 2}^{2^k - 1}$.
\end{proposition}

\begin{proof}
A $(2,k)$-tree has $2^{k} - 1$ non-leaf nodes. Every non-leaf node has ${d \choose 2}$ ways of choosing two children different from its parent node.
\end{proof}
\medskip
\begin{proposition}
\label{pro:expectedW}
Let $v$ be the root of a $(d,k)$-tree $T$ in $G$. Suppose we activate every vertex in $G$ independently with probability $p$. Then the expected number of $k$-witnesses for $v$ in $T$ is at least ${d \choose 2}^{2^k - 1}p^{2^k}$.
\end{proposition}

\begin{proof}
By Proposition~\ref{pro:2ind} the number of $(2,k)$-trees rooted at $v$ that $T$ contains is ${d \choose 2}^{2^k - 1}$. Each one of them has $2^k$ leaves, and all its leaves are seeds with probability $p^{2^k}$ if these leaves are distinct, and higher probability otherwise.
\end{proof}
\medskip
To show that a vertex $v$ is likely to be activated, we shall view it as a root of a $(d,k)$-tree, and show that this tree is likely to contain a $(2,k)$-witness for $v$. A necessary condition for this is that the expected number of $(2,k)$-witnesses will exceed~1. By Proposition~\ref{pro:expectedW}, this will happen when $p > d^{\frac{2}{2^k} - 2}$. To make this into a sufficient condition, we develop tools for bounding the variance of this random variable.

\begin{definition}
A $(d,k)$-tree $T$ in a graph $G$ is {\em proper} if all its nodes correspond to distinct vertices of $V$. Equivalently, the subgraph of $G$ induced by the edges of $T$ does not contain a cycle. The tree $T$ is $t$-proper if the subgraph of $G$ induced by the edges of $T$ does not contain a $t$-cycle in $G$. (Edges of $T$ that correspond to the same edge in $G$ are counted only once.)
\end{definition}

\begin{proposition}
\label{pro:4cycle}
Let $G$ be a graph with no 4-cycles. Then every $(d,k)$-tree in $G$ is 4-proper.
\end{proposition}

\begin{proof}
By definition.
\end{proof}

\begin{lemma}
\label{lem:2tree}
Let $v$ be the root of a 4-proper $(d,2)$-tree $T$, and let $\omega(d^{-2}) \le p \le o(d^{-3/2})$. Then the probability that $v$ has a 2-witness in $T$ is at least $(1 - o(1)){d \choose 2}^3p^4$.
\end{lemma}

\begin{proof}
All $d^2$ leaves in $T$ are distinct, because $T$ is 2-proper. Let $W_i$ denote the indicator random variable for the event that the $i$th $(2,2)$-tree in $T$ is a 2-witness for $v$. Then $Pr[W_i=1] = p^4$. Let $W = \sum W_i$ be a random variable that counts the number of 2-witnesses in $T$ for $v$. Then $E[W] = {d \choose 2}^3p^4$ (which is the same as substituting $k=2$ in Proposition~\ref{pro:expectedW}).

Consider an arbitrary $(2,2)$-tree in $T$, and suppose that it happens to be a witness. W.l.o.g we can assume $T$ is the $i$th tree, that is, $W_i=1$ (all its leaves are seeds). We compute an upper bound on $E[W|W_i=1]$.
Hence conditioned on $W_i=1$, we only know of four leaves that are seeds.
The number of $(2,2)$-trees that share three leaves with $T$ is $4(d-2)$ (each of the four leaves of $W_i$ can be replaced by $d-2$ alternative leaves). The number of $(2,2)$-trees that share two leaves with $T$ is at most $2(d-2){d \choose 2} + 4(d-2)^2$ (either one of the two children of $v$ in $W_i$ is replaced by a different child with two leaves, or each of the children of $v$ has one of its leaves replaced). The number of $(2,2)$-trees that share one leaf with $T$ is at most $4(d-2)^2{d \choose 2}$ (one child of $v$ replaces a leaf, and another child of $v$ is placed completely).  Hence
$$E[W|W_i=1] \le 1 + 4dp +  (d^3 + 4d^2)p^2 + 2d^4p^3 + {d \choose 2}^3p^4 \le 1 + O(d^3p^2),$$
where the last inequality used $\omega(d^{-2}) \le p \le o(d^{-3/2})$.
It follows that
$$E[W^2] = \sum_i Pr[W_i]E[W|W_i=1] \le (1 + O(d^3p^2))\sum_i Pr[W_i=1] = (1 + O(d^3p^2))E(W).$$
Observe that by definition $1 = \sum_{i=0}^{\infty} Pr[W=i]$, that $E[W] = \sum_{i=0}^{\infty} iPr[W=i]$, and that $E[W^2] = \sum_{i=0}^{\infty} i^2Pr[W=i]$. Hence (see \cite{Boll}, Theorem 1.16)
$$Pr[W = 0] \le 1 - 2E[W] + E[W^2] \le 1 - (1 - O(d^3p^2))E(W),$$
implying  that $Pr[W > 0] \ge (1 - O(d^3p^2))E(W) = (1 - o(1)){d \choose 2}^3p^4 \ge \Omega(d^6p^4)$.
\end{proof}
\medskip
\begin{lemma}
\label{lem:3tree}
Let $v$ be the root of a 4-proper $(d,3)$-tree $T$. Then $v$ has probability at least $1/2$ of being activated when $p = 4d^{-7/4}$. (The leading constant~4 was chosen for concreteness. A smaller constant suffices.)
\end{lemma}

\begin{proof}
Let $v_1, \ldots, v_d$ denote the neighbors of $v$ in $T$. Let $X_i$ be an indicator random variable for the event that $v_i$ has a 2-witness in the subtree of $T$ rooted at $v_i$. Lemma~\ref{lem:2tree} implies that $Pr[X_i=1] = (1 - o(1)){d \choose 2}^3p^4$. Let $X = \sum X_i$. Then $E[X] = (\frac{1}{8} - o(1))d^7p^4 \simeq 32$. Observe that when $X \ge 2$ at least two neighbors of $v$ are activated, and then $v$ is activated as well. Hence if $X$ behaves similar to its expectation, we expect $v$ to be activated. To show that $X$ is concentrated around its expectation, we compute $E[X^2]$.

Let us compute $Pr[X_i \wedge X_j]$ for $i \not= j$. The fact that $T$ is 4-proper implies the following useful facts:

\begin{enumerate}

\item All $d^2$ leaves in the subtree of $T$ rooted at $v_i$ are distinct. All $d^2$ leaves in the subtree of $T$ rooted at $v_j$ are distinct.

\item All children of $v_i$ in $T$ are distinct from all children of $v_j$ in $T$.

\item No child of  $v_i$ in $T$ has two common children with a child of $v_j$ in $T$.

\end{enumerate}

The probability $Pr[X_i \wedge X_j]$ depends on the pattern of common grandchildren that the vertices $v_i$ and $v_j$ has. The above facts show that every child of $v_i$ and every child of $v_j$ have at most one common neighbor. We consider two cases.

In the first case every child of $v_i$ and every child of $v_j$ have {\em exactly} one common neighbor. This case can be visualized as a $d$ by $d$ table $M$ of distinct grandchildren. The rows are indexed by the children of $v_i$ and the columns are indexed by the children of $v_j$. Every child of $v_i$ is a neighbor of those grandchildren in its respective row, and every child of $v_j$ is a neighbor of those grandchildren in its respective column. Each entry of the table is a seed with probability $p$ and not a seed otherwise. For the event $X_i \wedge X_j$ we need two rows to have two seed entries, and two columns to have two seed entries. This requires between four to eight seed entries, depending on where the seeds are located within the table. We compute the number of possibilities for each case separately.

\begin{enumerate}

\item Four seed entries. One needs to choose the two rows and two columns that contain them, giving ${d \choose 2}^2$ possibilities.

\item Five seed entries. There are $\Theta(d^6)$ possibilities. (Details omitted.)

\item Six seed entries. There are $\Theta(d^8)$ possibilities. (Details omitted.)

\item Seven seed entries. There are $\Theta(d^{10})$ possibilities. (Details omitted.)

\item Eight seed entries. One needs to choose two rows and two locations within these rows, and likewise for the columns. This gives at most ${d \choose 2}^6$ possibilities.

\end{enumerate}

As $d^2p \gg 1$, the dominating term is ${d \choose 2}^6p^8$, giving $Pr[X_i \wedge X_j] = (1 + o(1))Pr[X_i]Pr[X_j]$. It follows that
$$E[X^2] = \sum_i \sum_j  Pr[X_i \wedge X_j] \le \sum_i Pr[X_i](1 + (1 + o(1))E[X]) = E[X] + (1 + o(1))(E[X])^2,$$
Hence $var[X] = E[(X - E[X])^2] = E[X^2] - (E[X])^2 = E[X] + o((E[X])^2)$.

Now Chebyschev's inequality implies that $Pr[X \ge 2] > 1/2$.

The remaining case to consider is the one in which some pairs of children, one child of $v_i$ and one child of $v_j$, have no common neighbors at all. In this case, some entries of the table $M$ referred to above are empty, and instead the vertices representing the corresponding rows and columns have additional children not accounted for in $M$ (and not shared by other vertices). Imitating the analysis performed for the first case, the number of possibilities for eight seed entries remains at most ${d \choose 2}^6$, and ${d \choose 2}^6p^8$ remains the dominating term (the nondominating terms can easily be seen not to increase by more than a constant factor). Hence the bounds proven for the first case above apply also in the current case.
\end{proof}
\medskip
We can now prove Theorem~\ref{thm:4cycle}.

\begin{proof}
As the minimum degree of $G$ is $d+1$, every vertex $v$ in $G$ is a root of a $(d,3)$-tree. By Proposition~\ref{pro:4cycle} this $(d,3)$-tree is proper. By Lemma~\ref{lem:3tree}, if $p = 4d^{-7/4}$ then $v$ is activated with probability at least~$1/2$. By Corollary~\ref{cor:partialinfection}, there is a contagious set of size $2pn$.
\end{proof}
\medskip

\section{Contagious sets in graphs of girth at least 7}
\label{sec:girth9}

Before proving Theorem~\ref{thm:girth7}, let us present a lemma that summarizes the only property of $d$-regular graphs of girth at least~7 that will be used in the proof. Given a graph $G(V,E)$, for a set $S$ of vertices, recall that $N(S)$ denote the set of those vertices that are neighbors of some vertex in $S$, and let $N^2(S)$ denote the set of those vertices that are at distance exactly~2 from some vertex in $S$. Observe that we do not require the sets $S$, $N(S)$ and $N^2(S)$ to be disjoint.

\begin{lemma}
\label{lem:girth7}
Let $G$ be a $d$-regular graph of girth at least~7. Then for every $1 \le k < d$ and every set $S$ of $k$ vertices it holds that $|N^2(S)| \ge \frac{kd^2}{2}$.
\end{lemma}

\begin{proof}
Given a $d$-regular graph $G(V,E)$ of girth at least~7, consider an arbitrary set $S$ of $k$ vertices. For every vertex $v\in S$ we have that $|N^2(v)| = d(d-1)$, because otherwise $G$ has a cycle of length at most~4. Hence $\sum_{v\in S} |N^2(v)| = kd(d-1)$. To provide a lower bound on $|N^2(S)|$, we use the first two terms of the inclusion exclusion formula. Namely:

$$|N^2(S)| \ge  kd(d-1) - \sum_{u,v \in S} |N^2(u) \cap N^2(v)|$$

We now claim that for every $u,v \in V$ it holds that $|N^2(u) \cap N^2(v)| \le d$. Suppose otherwise that $|N^2(u) \cap N^2(v)| > d$. Then by the pigeon-hole principle, and least one vertex $x \in N(u)$ has at least two neighbors $x_1,x_2$ in $N^2(v)$. Suppose first that $x \not\in N(v)$. Then $x_1$ and $x_2$ cannot have a common neighbor $y$ in $N(v)$, because then $x,x_1,y,x_2$ would form a 4-cycle. Hence there are two vertex disjoint paths from $x$ to $v$ (one through $x_1$, the other through $x_2$). This forms a 6-cycle, which contradicts the girth assumption.

The other case to consider is that $x \in (N(u) \cap N(v))$. (Note that it cannot be that $x = v$ because in that case neighbors of $x$ will not be in $N^2(u) \cap N^2(v)$.) Observe that then there cannot be any other vertex $y$ that is in $N(u) \cap N(v)$, because $x,u,y,v$ would form a 4-cycle. Observe also that $|N^2(u) \cap N^2(v)| > d$ implies that there is a vertex $z \not\in N(x)$ that is in $N^2(u) \cap N^2(v)$. This $z$ has two vertex disjoint paths of length~3 to $x$, one through $u$ and the other through $v$. This forms a 6-cycle, contradicting the girth assumption.

If follows (using also $k < d$) that:

$$|N^2(S)| \ge  kd(d-1) - d{k \choose 2} = kd(d - 1 - \frac{k-1}{2}) \ge \frac{kd^2}{2}$$
\end{proof}

{\bf Remark.} The proof of Lemma~\ref{lem:girth7}  only requires the graph not to have 4-cycles and 6-cycles. Having arbitrarily short odd cycles does not matter, up to some minimal changes in the parameters, such as the allowed range of $k$, or the leading term of $\frac{1}{2}$ for the expression $kd^2$. Consequently, the proof of Theorem~\ref{thm:girth7} only uses the absence of 4-cycles and 6-cycles, and not the full requirement of girth at least~7. More generally, existence of odd cycles can have only limited effect on upper bounds on $m(G,2)$, as long as these upper bounds are expressed as function of the degree and do not require the graph being exactly regular. This can be seen by recalling that every $d$-regular graph has a maximal cut in which every vertex has between $d/2$ and $d$ edges crossing the cut. Removing all edges except for cut edges leaves us with a bipartite graph $G'$, which has no odd cycles. Furthermore, all degrees are between $d/2$ and $d$. Upper bounds on $m(G',2)$ trivially apply to $G$ as well. Finally, observe that Lemma ~\ref{lem:girth7} is no longer true if we only require the graph to have no four-cycles (or girth 5) as there are $d$ regular graphs with girth $5$ and $O(d^2)$ vertices.

We now prove Theorem~\ref{thm:girth7}.

\begin{proof}
We present an algorithm that is partly random and partly greedy for selecting a contagious set in $G(V,E)$. Let $p = \frac{4\ln d}{d^2}$. Let $A$ be an initial set of seeds, where every vertex of $G$ in included in $A$ independently at random with probability $p$.
Given $A$, consider the following sets of vertices.

\begin{enumerate}

\item Set $A$ of seeds.

\item Set $B$ of excited vertices: vertices in $V \setminus A$ that have at least one neighbor in $A$. Observe that under our definition of $B$, a vertex in $B$ may have two or more neighbors in $A$ and hence be activated, but we still refer to it as excited. Consider the subgraph $G(B)$ of $G$ induced on the vertices of $B$. Call a connected component in $G(B)$ {\em large} if it contains at least $d$ vertices, and {\em small} otherwise. Based on this distinction, we partition $B$ into two disjoint subsets.

    \begin{enumerate}

    \item The set $B_{L}$ of vertices that are in large connected components in $G(B)$.

    \item The set $B_{S}$ of vertices that are in small connected components in $G(B)$.

\end{enumerate}

\item Set $C$ of those vertices in $V \setminus (A \cup B)$ that have at least one neighbor in $B_{L}$.

\end{enumerate}

As a memory aid, one may think of $A$ as representing {\em activated}, $B$ as representing {\em boundary}, and $C$ as representing {\em close}.

Consider an arbitrary vertex $v\in V$. We analyze the probability of the event that $v \in C$. This event can be broken into several other events that all need to happen simultaneously.

{\bf Event $\bar{A}_v$}, which holds if $v \not\in A$. This happens with probability $1 - p$.

{\bf Event $\bar{B}_v$}, which holds if $v \not\in B$. This happens with probability at least $1 - dp$, because $v$ has $d$ neighbors.

{\bf Event $NB_v$}, which holds if $v$ has at least one neighbor in $B$. Consider the vertices at distance~2 from $v$. As $G$ has no 4-cycles, these are $d(d-1)$ distinct vertices. The expected number of these vertices that are in $A$ is $pd(d-1) \simeq 4\ln d$. Hence the probability that at least one of them is in $A$ is roughly $1 - e^{-4\ln d} > 1 - \frac{1}{d}$. Let $w \in A$ be a vertex at distance~2 from $v$, and let $u$ be the common neighbor of $v$ and $w$. If $u$ is not in $A$ (which happens with probability $1 - p$) then $u$ is in $B$.
    Hence event $NB_v$ holds with probability at least $1 - \frac{1}{d} - p \ge 1 - \frac{2}{d}$.

{\bf Event $\overline{NS}_v$}, which holds if $v$ has no neighbor in $B_S$.

\begin{lemma}
\label{lem:violating}
The Event $\overline{NS}_v$ holds with probability $1 - O(1/d)$.
\end{lemma}

\begin{proof}
Consider an arbitrary vertex $u\in N(v)$, and for $k < d$, let $K$ be a connected set of $k$ vertices that contains $u$. Consider the event $\bar{K}$ that $K$ forms one of the connected components in $B$. This event involves two requirements: one is that $K \subset B$ and the other is that no vertex in $\partial(K)$ is in $B$. Observe that by considering all possible connected $K$ that contain $u$, exactly one of the events $\bar{K}$ needs to happen in order for $k$ to be the size of the connected component of $u$ in $G(B)$. Given that $G$ is of degree $d$ and that $u \in K$, there are at most ${(k-1)d \choose k-1} \simeq (ed)^{k-1}$ ways of choosing the $k$ vertices of $K$.

Given $K$, we now upper bound the probability of event $\bar{K}$. For this, it suffices to upper bound the probability that no vertex in $\partial(K)$ is in $B$ (while ignoring the requirement that $K \subset B$). This event fails if a vertex $z$ at distance~2 from a vertex of $x\in K$ is {\em violating}, namely, $z\in A$, and there is a vertex $y \in  N(x) \cap N(z)$ such that $y \not\in (A \cup K)$. This $y$ is in $B$ and can be used to enlarge $K$.
Lemma~\ref{lem:girth7} implies that $N^2(K) \ge \frac{kd^2}{2}$. Using this, we now estimate the probability that no violating vertices exist.

For every vertex $z \in N^2(K)$, designate one vertex in $N(z) \cap N(K)$ to be the {\em link} $l(z)$ to $K$. Observe that every vertex in $N(K)$ can serve as a link to at most $d$ vertices in $N^2(K)$ (because the graph has degree $d$). At most $k$ of the links are in $K$ ($N(K)$ may not be disjoint from $K$). Ignore those vertices in $N^2(K)$ whose link is in $K$. This still leaves at least $\frac{kd^2}{2} - kd$ vertices in $N^2(K)$ whose link is not in $K$. With each link $l$ that is not in $K$, associate a 0/1 random variable $y_l$ whose value is~1 if and only if the following two conditions hold: $z \in A$ for at least one $z \in N^2(K)$ for which $l(z) = l$, and $l \not\in A$. Let $d_l \le d$ denote the number of $z \in N^2(K)$ for which $l(z) = l$. We get that $Pr[y_l = 1] \ge (1-p)d_lp(1-p)^{d_l-1} \simeq p d_l$ (where the near equality holds because for our choice of $p$ and $d$, $(1 - p)^d \simeq 1$). If $y_l = 1$ then there is a violating vertex. Let $Y = \sum_l y_l$. There is no violating vertex only if $Y = 0$. Note that the expectation of $Y$ is roughly $\sum pd_l \ge p(\frac{kd^2}{2} - kd) \simeq 2k\ln d$. Observe that the random variables $y_l$ are independent, and each of them is a 0/1 variable, hence standard concentration results imply that $Pr[Y = 0] \le e^{-2k\ln d} \simeq d^{-2k}$.

Taking a union bound over all choices of $K$, it follows that the size of the connected component of $u$ in $G(B)$ is exactly $k$ with probability at most $d^{-2k}(ed)^{k-1} \le (\frac{e}{d})^{k+1}$. Summing over all values of $1 \le k < d$, the probability that $u\in B_S$ is $O(1/d^2)$. Taking a union bound over all neighbors of $v$, we get that $Pr[\overline{NS}_v] = 1 - O(1/d)$.
\end{proof}

For a given vertex $v$, if all four events listed above hold simultaneously then $v \in C$ (observe that the combination of $NB_v$ and $\overline{NS}_v$ imply that $v$ has a neighbor in $B_L$). Hence $v \in C$ with probability at least $1 - p - pd - O(1/d) > 3/4$ (for our choice of $p$ and sufficiently large $d$).

Within every large component (in $B_L$), chose at random one vertex to be a seed. Observe that the probability that $v$ becomes a seed by this is at most $p$ (probability of $pd$ for being in $B$, times probability at most $1/d$ of being selected as seed in his large component). Observe also that this activates the whole large component. Hence by now every vertex of $C$ has at least one active neighbor.

Let us repeat the above experiment of selecting a random $A$ twice, each time with fresh randomness. Call a vertex {\em lucky} if it is in $C$ in both experiments. Hence the probability that a vertex $v$ is lucky is at least $(\frac{3}{4})^2 = \frac{9}{16}$. If the two active neighbors of $v$ are distinct, then $v$ is infected as well. What is the probability that these two active neighbors are not distinct? For this, $v$ would have to have a neighbor that is in $B$ in both experiments. This happens with probability at most $d(pd)^2 \le \frac{1}{16}$ (for our choice of parameters). Hence $v$ has probability at least $1/2$ of becoming infected. Note also that $v$ had probability at most $4p$ of becoming a seed in at least one of the experiments. Hence Lemma~\ref{lem:partialinfection} implies that $G$ has a contagious set of size $8pn = O(\frac{n\log d}{d^2})$.
\end{proof}

\section{Contagious sets in expanders with no 4-cycles}

In this section we prove Theorem~\ref{thm:girth_expander}.

Our strategy in building a small contagious set for expanders with no 4-cycles will be to choose the seeds (the vertices we activate) one by one in rounds in a greedy manner, where for a given round $t$, $s_t$ will denote the seed chosen in round $t$, and $S_t$ will denote the set of all $t$ seeds chosen up to and including round $t$. Given a set $S_t$ of seeds, an activation cascade may activate additional vertices. We let $A_t$ denote the set of all activated vertices after round $t$, with $S_t \subset A_t$. We shall be concerned also with neighbors of vertices in $A_t$, and denote $B_t = A_t \cup \partial(A_t)$. The set of remaining vertices in $V \setminus B_t$ will be denoted by $R_t$. Initially, $S_0$, $A_0$ and $B_0$ are empty, and $R_t = V$.

Our greedy algorithm has two phases, each employing a different greedy rule. It switches between phases once $B_t$ becomes the majority of the graph. Specifically, at round $t\ge 1$, if $A_{t-1} \not= V$, the greedy algorithm proceeds as follows:

\begin{enumerate}

\item If $|B_{t-1}| < n/2$, select as seed $s_t$ a vertex $v \in (V \setminus A_{t-1})$ such that $S_t = S_{t-1} \bigcup \{v\}$ maximizes $|B_t|$ (after applying the activation cascade).

\item If $|B_{t-1}| \ge n/2$, select as seed $s_t$ a vertex $v \in (V \setminus A_{t-1})$ such that $S_t = S_{t-1} \bigcup \{v\}$ maximizes $|A_t|$ (after applying the activation cascade).

\end{enumerate}

We let $T$ denote the total number of rounds until $A_T = V$. We now establish that $T = O(\frac{n \log d}{\epsilon^2 d^2})$.  The following lemma does not require any expansion properties.

\begin{lemma}
\label{lem:FirstPhase}
Let $G(V,E)$ be an arbitrary $d$-regular graph. Then for $t \le \frac{n}{d}$ the above greedy algorithm can maintain $|B_t| \ge \frac{d}{2}t$.
\end{lemma}

\begin{proof}
By induction on $t$. For $t=1$ we have $A_1 = \{s_1\}$ and hence $|A_1| = 1$, $\partial(A_1) = d$, and $|B_1| = d + 1 \ge d/2$. Assume now that the lemma holds for $t < \frac{n}{d}$ and prove for $t+1$. If $|B_t| \ge \frac{d}{2}(t+1)$ there is nothing to prove. Hence we may assume that $|B_t| < \frac{d}{2}(t+1) \le n/2$, implying that $|R_t| \ge n/2$. Therefore $\sum_{v \in V} deg_{R_t}(v) \ge \frac{nd}{2}$, and a random vertex has in expectation at least $d/2$ neighbors in $R_t$. Hence there is at least one vertex $v$ with at least $d/2$ neighbors in $R_t$. It cannot be that $v \in A_t$ because vertices in $A_t$ have no neighbors in $R_t$. Hence taking this vertex $v$ as $s_t$ we have $|B_{t+1} \setminus B_t| \ge d/2$, proving the inductive step.
\end{proof}
\medskip

The weakness of Lemma~\ref{lem:FirstPhase} is that the rate of growth of $B_t$ is limited to $O(dt)$. To reach $B_T$ linear in $n$ will require $T \ge \Omega(n/d)$, which we cannot afford. Hence we shall want to establish that $B_t$ grows at a rate significantly larger than $d$ per round. This is clearly not true in the first set of rounds (in particular, $|B_1|=d+1$), but we shall show that it becomes true after $t$ exceeds $n/d^2$. Our next lemma does use expansion properties of $G$.

\begin{lemma}
\label{lem:SecondPhase}
For $0 < \epsilon < 1$, let $G(V,E)$ be an $(n,d,\lambda)$-graph with $\lambda \le (1 - \epsilon)d$ and without 4-cycles.  Let $\frac{4}{\epsilon^2 d} \le c \le \frac{d}{2}$. Let $A$ be an arbitrary set of activated vertices in $G$, let $B=A \cup \partial(A)$ and let $R = V \setminus B$. If $|B| = \frac{cn}{d}$  then there is a vertex $u \in R$ such that $|R \cap N(B \cap N(u))| \ge c\epsilon^2 d/2$.
\end{lemma}
\medskip

\begin{proof}
Three vertices $u,v \in R$ and $w \in B$ will be called a {\em triplet} if $(u,w) \in E$ and $(v,w)\in E$. Let $f$ denote the number of triplets in $G$. For $w \in B$, let $d_R(w) = |N(w) \cap R|$. Then $f = \sum_{w \in B} {d_R(w) \choose 2}$. Using Lemma~\ref{lem:expansion}, $\sum_{w \in B} d_R(w) = e(B,R) \geq \frac{\epsilon d |B||R|}{n}$. Hence the average value of $d_R(w)$ is at least $\frac{\epsilon d |R|}{n}$, implying by convexity that

$$f \ge |B|{\epsilon d |R|/n \choose 2} \simeq \frac{|R|^2|B| \epsilon^2 d^2}{2n^2}.$$

Every triplet involves two vertices from $R$. Hence on average, a vertex from $R$ is involved in $2f/|R|$ triplets. This together with the lower bound on $f$ implies that there is some $u \in R$ involved in at least $\frac{|R| |B| \epsilon^2 d^2}{n^2}$ triplets. In any two such triplets, $(u,w_1,v_1)$ and $(u,w_2,v_2)$ ($v_1,v_2 \in R$), the vertices $v_1$ and $v_2$ must be distinct, because $G$ has no 4-cycles. This implies that $|R \cap N(B \cap N(u))| \ge \frac{|R| |B| \epsilon^2 d^2}{n^2}$. Substituting $|B| = cn/d$ and noting that $|R| \ge n/2,$ the lemma follows.
\end{proof}
\medskip

We now proceed to prove Theorem \ref{thm:girth_expander}:

\begin{proof}
Lemma~\ref{lem:FirstPhase} implies that for $t = \frac{8n}{\epsilon^2d^2}$ the greedy algorithm reaches $|B_t| \ge \frac{4n}{\epsilon^2 d}$. Thereafter, in every $O(\frac{n}{\epsilon^2 d^2})$ iterations of the algorithm, Lemma~\ref{lem:SecondPhase} implies that $B_t$ grows by a multiplicative factor of $2$ (in every iteration choose the vertex $u$ whose existence is guaranteed by Lemma~\ref{lem:SecondPhase}). It follows that for $T \le O(\frac{n \log d}{\epsilon^2 d^2})$ the greedy algorithm manages to achieve $|B_T| \ge \frac{n}{2}$, and the first phase of the greedy algorithm ends.

We now analyze the second phase of the greedy algorithm. We may assume that $|A_t| \le \frac{n}{\epsilon d}$, because otherwise the whole graph is activated, by Lemma \ref{lem:general_g}. Moreover, we may assume that $d > \frac{10}{3\epsilon}$, as otherwise the statement of Theorem \ref{thm:girth_expander} only requires $m(G,2) \le O(n \log d)$ which is trivially true. For this range of parameters, $|\partial(A_t)| = |B_t| - |A_t| \ge \frac{n}{2} - \frac{n}{\epsilon d} > \frac{2n}{5}$. Each vertex in $\partial(A_t)$ has exactly one neighbor in $A_t$, and hence $e(\partial(A_t), V \setminus A_t) \ge (d-1)\frac{2n}{5} \ge \frac{dn}{3}$. This implies that there is some vertex in $V \setminus A_t$ whose activation will activate at least $d/3$ new vertices. Hence the greedy algorithm activates at least $d/3$ vertices in each step of the second round, implying that in $O(\frac{n}{\epsilon d^2})$ rounds of the second phase $|A_t|$ exceeds $\frac{n}{\epsilon d}$. Lemma \ref{lem:general_g} then implies that the whole graph is activated.
\end{proof}

\section{Bounds for contagious sets in random graphs}

In this Section we prove Theorem~\ref{thm:random}, which is a direct consequence of Theorems~\ref{thm:upper_bound} and~\ref{thm:LB_random}.

Let $G:=G(n,p)$ be the binomial random graph over $n$ vertices and edge probability $p$ and assume that $\frac{3 \log n}{n}<p<\frac{1}{n^{\frac{1}{2}+\epsilon}}$ for some fixed $\epsilon \in (0,1)$ (it is not hard to see that if $p>\frac{1}{n^{1/2-\epsilon}}$ then $m(G,2)=2$ with high probability). Let $d:=np$. Janson, {\L}uczak, Turova and Vallier \cite{JanLuc} prove that with w.h.p. $m(G,2) \leq \frac{(1+\delta)n}{2d^2}$ for every $\delta>0$ (Theorem 3.1, page 1996). Furthermore, a given set $A$ of cardinality at least $\frac{(1+\delta)n}{2d^2}$ infects the entire graph with high probability, once activated (page 1990, one before last paragraph).
\subsection{Upper bound}
Here we show that w.h.p. $m(G,2) \le O(\frac{n \log \log d }{d^2\log d})$. We use the following lemma.
\begin{lemma}
\label{lemma:subcritical}
Let $H:=G(n_0,q)$ be the binomial random graph with $n_0$ vertices and edge probability $q$. Assume $q=\frac{c}{n_0}$ and $c<1/10$ ($c$ may depend on $n_0$).
Let $k=O(\log n_0)$ be an integer and $v \in H.$ Then the probability $v$ belongs to a connected component of size at least $k$ is at least $c^{3k}$. Furthermore, w.h.p. the number of vertices in components of size at least $k$ is at least $c^{3k}\cdot n_0/4$.
\end{lemma}
\begin{proof}
Given a vertex $v$ in $H$, expose the connected component of $v$ in $H$ using in breadth first search (BFS) manner until either one of two cases occur: a {\em success}, meaning that the size of the connected component containing $v$ revealed by the BFS algorithm reaches $k$, or a {\em failure}, meaning that the BFS algorithm dies out before accumulating $k$ vertices in the connected component of $v$. Next, proceed by discarding the connected component containing $v$ that was just revealed from $H$. Repeat this process until the number of remaining vertices is smaller than $n_0/2$. As long as we continue to reveal connected components, the number of vertices that were not discarded thus far is at least $n_0/2$. It follows that the probability $v$ is in a component of size at least $k$ is at at least
%(for simplicity assume that $k-1$ is a power of $2$)
$$\left({n_0/2\choose 2}\left(\frac{c}{n_0}\right)^2\left(1-\frac{c}{n_0}\right)^{n_0/2}\right)^{k-1}\geq \left(\frac{c^2}{8}e^{-\frac{1}{c}}\right)^{k-1}.$$
Since $c<1/10$, we can lower bound the expression above by $c^{3k}$.
The distribution of the number of successes (until less than $n_0/2$ vertices remain) stochastically dominates the binomial distribution with $\lfloor n_0/(2k) \rfloor$ trials and success probability $c^{3k}$ (the exact number of trials depends on the number of failures, but failures only increase the number of trials). The actual number of successes is concentrated around its expectation, a fact that can be proved using standard concentration results for martingales (further details omitted). This combined with the fact that every success places $k$ vertices (rather than just one) in a component of size at least $k$ implies the lemma.
\end{proof}

We now use the above lemma to prove that we can activate ``many" vertices using $O(\frac{n\log \log d}{d^2 \log d})$ vertices.
\begin{lemma}
\label{lemma:seeds}
Let $G:=G(n,p)$ and let $d = pn$. Then with high probability, one can infect $\frac{n}{d^2}$ vertices by activating $O\left(\frac{n\log \log d}{d^2 \log d}\right)$ seeds.
\end{lemma}
\begin{proof}
Activate an arbitrary set $A$ of $\frac{\epsilon n}{d^2}$ vertices where where $\epsilon:=\frac{\log\log d}{\log d}$.
It is not hard to verify that $\partial(A)$, the set of all neighbors of $A$ not in $A$, satisfies with high probability  $$\frac{\epsilon n}{2d} \le |\partial(A)| \le \frac{2\epsilon n}{d}.$$ The induced graph on $\partial(A)$ is distributed as $G(n_0,p)$ with $n_0:=|\partial(A)|$. Hence $p = \frac{d}{n} \ge \frac{\epsilon}{2 n_0}$. Hence Lemma~\ref{lemma:subcritical} (with $c$ taken to equal $\frac{\epsilon}{2}$) implies that the number of vertices in $\partial(A)$ lying in components of size at least $k$ is with high probability at least $\left(\frac{\epsilon}{2}\right)^{3k}\cdot\left(\frac{n_0}{4}\right)$. Setting $k=\frac{\log d}{6\log \log d}$, we get that $\left(\frac{\epsilon}{2}\right)^{3k}\geq \frac{8}{\epsilon d}$. Hence, w.h.p. at least $\frac{2n_0}{\epsilon d} \ge \frac{n}{d^2}$ vertices in $\partial(A)$ lie in components of size at least $k$. Activating a single vertex in every such component will result with an active set of size at least $\frac{n}{d^2}$, whereas the total number of activated vertices is $$\frac{\epsilon n}{d^2}+\frac{n}{k \cdot d^2}=O\left(\frac{n\log \log d}{d^2 \log d}\right).$$
The lemma follows.
\end{proof}

We can proceed and prove the main result of this section.
\begin{theorem}
\label{thm:upper_bound}
Let $G:=G(n,p)$ where $\frac{3 \log n}{n}<p<\frac{1}{n^{{\frac{1}{2}}+\epsilon}}$. Then with high probability $m(G,2)\leq O(\frac{\log\log (np)}{np^2 \log (np)})$.
\end{theorem}
\begin{proof}
Recall that we define $d$ to equal $np$. By Lemma \ref{lemma:seeds} we can activate a set $I$ such that $|I|=\frac{n}{d^2}$ by first activating a set $A$ of seeds and then activating an additional set of seeds in $A' \subseteq \partial(A)$ where $|A|+|A'| = O(\frac{n\log \log d}{d^2 \log d})$ and furthermore $I \subseteq \partial(A)$. Let $G'$ be the graph induced on $(V\setminus (A \cup \partial(A))) \cup I$. Then the edges of $G'$ except for those induced by $I$ are distributed as $G(l,p)$ where with high probability $l>(1-\frac{1}{d})n$ (as $|\partial (A)| = o(\frac{n}{d})$). Observe that in our activation procedure, we have not revealed any information about any edge in $G'$ other than edges with both endpoints in $I$ (here it is important that $I$ is disjoint from $A$). As $|I| \geq \frac{n}{d^2}$, the result of~\cite{JanLuc} implies that with high probability $I$ infects the whole of $G'$ (note that since $I$ is activated, the pattern of edges with $I$ is irrelevant, and hence the results of~\cite{JanLuc} apply). Let $B:=N(A) \setminus I$. It remains to prove that also $B$ is infected. The probability a vertex in $B$ does not have a neighbor in $G'$ is at most $(1-p)^{n/2}\sim e^{-pn}=o(1/n)$, as $p>\frac{3\log n}{n}$. Hence using the union bound, with high probability, every vertex in $B$ has a neighbor in $G'$, which as we just proved, is activated. Moreover, every vertex in $B$ has a neighboring seed in $A$ (by definition of $B$), and hence has at least two activated neighbors. Thus it becomes active as well.
\end{proof}

\subsection{Lower bound}
In this section we prove a lower bound for $m(G,2)$ on $G \sim G(n,p)$ where $p$ is as in the previous section. (In fact, the proof of the lower bound applies virtually without change for all $p > 2/n$, though for such small values of $p$ there are simpler ways of proving similar bounds, for example, by counting isolated vertices.)

For the lower bound let us recall a few observations made in \cite{JanLuc}.
Suppose we activate an initial set $A$ with $|A|=a>0$ vertices in $G$. Now we track how vertices outside $A$ become infected as follows. Throughout we record \emph{active} vertices and \emph{used} vertices. In the beginning, all vertices in $A$ are active, and the set of used vertices is empty. In each iteration, we choose an active vertex $v$ (provided the set of active vertices is nonempty), expose all edges between $v$ and all vertices which are not labeled as active or used presently and add a mark to all adjacent vertices to $v$. Thereafter $v$ is now tracked as "used" and all nonactive vertices that become active (have two marks) after inspecting all edges adjacent to $v$ are added the to set of active vertices.

Suppose this process runs for $t$ iterations. For a vertex $w$ not in $A$ and for a vertex $u$ that we considered in $i$th iteration, $w$ gets an  additional mark from $u$ (that is, $w$ is a neighbor of $u$) with probability $p$. Hence $w$ is activated (for threshold $r$) by time $t$ with probability
$$\pi(t):=\Pr({\rm Bin}(t,p) \geq r)=\sum_{j=r}^t{t \choose j}p^j(1-p)^{t-j}.$$
For $r=2$ and $t\ll\frac{1}{p}$  it can be verified (see \cite{JanLuc}) that $\pi(t) \approx\frac{(tp)^2}{2}.$

Let $A(t)$ be the number of active vertices at time $t<n$. Clearly the infection process will survive at time $t_0$ if and only if

  \begin{equation}
    \label{eq:first_cond}
    A(t)-t>0
  \end{equation}
for all $t<t_0$. In words, this condition means that for every $t<t_0$ the number of active vertices exceeds the number of used vertices. The number of vertices that are activated at time $t<t_0$ outside $A$ is distributed as $S(t):={\rm Bin}(n-a,\pi(t))$.
As $A(t)=S(t)+a$, Equation \ref{eq:first_cond} is equivalent to
 \begin{equation}
    \label{eq:second_cond}
 a+\min_{t<t_0}(S(t)-t)>0.
  \end{equation}
Let us now choose $t:=\frac{1}{np^2}$ (similarly to \cite{JanLuc}), $t_0:=t+1$ and $a:=\frac{1}{C\log (np)\cdot np^2}$ where $C>0$ is a large enough constant. Observe that by the choice of $p$, we indeed have that $t\ll\frac{1}{p}$.
In this setting $a-t\le-\frac{9}{10 np^2}$. Hence if we want the left hand side of \ref{eq:second_cond} to be positive we must have that
\begin{equation}
    \label{eq:S(t)}
    S(t)>\frac{9}{10 np^2}=\frac{9 n}{10 d^2}.
\end{equation}
The expectation of $S(t)$ is at most $\pi(t) n \simeq \frac{n(tp)^2}{2} = \frac{n}{2 d^2}$. Using the Chernoff bounds, we deduce that (\ref{eq:S(t)}) holds with probability at most $e^{-\Omega(-\frac{n}{d^2})}$. On the other hand, the number of all sets of size
$\frac{1}{C\log d \cdot np^2}$, is upper bounded by
$$((Ced^2\log d)^{\frac{1}{C\log d}})^{\frac{n}{d^2}}.$$
Taking $C$ large enough, and applying the union bound, gives us that with probability $1 - o(1)$ there is no set of size $\frac{1}{C\log d\cdot np^2}$ for which (\ref{eq:first_cond}) holds at time $t$.
We thus obtain:
\begin{theorem}
\label{thm:LB_random}
Let $G\sim G(n,p)$ with $\frac{3 \log n}{n}<p<\frac{1}{n^{\frac{1}{2}+\epsilon}}$. Then for large enough $C>0$ with high probability
$$m(G,2) \ge \frac{1}{C \cdot np^2\log (np)}$$
\end{theorem}

\medskip
\section{Bounds for $m(G,r)$: $r>2$}
\label{sec:bigthreshold}
In this section we give upper bounds for $m(G,r)$ where $r$ is a small constant (e.g., 3,4) not depending on $d$. The ideas are similar to Section~\ref{sec:spectral}, hence our proofs are less detailed.

\begin{lemma}
\label{lem:general_g,r}
Let $G$ be an $(n,d,\lambda)$-graph such that $\lambda<\delta d$ and $\delta<1$. Suppose that the activation threshold of every vertex is $r$ which is independent of $d$. Then every set of size larger than $\frac{(r-1)n}{(1-\delta)d}$ is contagious.
\end{lemma}
\begin{proof}
Consider a set $S$ of size $|S|$ that is not contagious. We can assume without loss of generality that $S$ is inclusion-maximal with respect to being active (namely, \emph{every} vertex \emph{not} belonging to $S$ is not active). For every $u \in V \setminus S$ it holds that $deg_S(u) \leq r-1$. Thus $e(S, V\setminus S) \leq (r-1)(n-|S|)$. On the other hand, by Lemma ~\ref{lem:expansion} $$e(S, V\setminus S)  \geq \frac{(1-\delta)d|S|(n-|S|)}{n}.$$
Combining these inequalities we have that
$$\frac{(1-\delta)d|S|(n-|S|)}{n} \leq (r-1)(n-|S|)$$
Hence $|S| \leq \frac{(r-1)n}{(1-\delta)d}$.
\end{proof}
\medskip
\begin{theorem} \label{thm:large_girth_r}
Let $G$ be a $d$-regular graph with girth $\Omega(\log\log d)$. Then there is a contagious set of size $C(r)nd^{-\frac{r}{r-1}}$ where $C(r)$ is a constant depending only on $r$.
\end{theorem}
\begin{proof}
The proof is similar to the proof of Lemma ~\ref{lem:tree}. Again, we consider $T_{d,k}$ the complete $d$-regular tree of depth $k$.
Recall that a vertex in $T_{d,k}$ is said to be in level $\ell$ with $0 \le \ell \le k$ if its distance from the root is $\ell$.
Activate all the leafs of $T_{d,k}$ independently with probability $h(k)d^{-\frac{r}{r-1}}$ where $h(k)=(2 e \cdot r!)^{\frac{2}{r}} d^{\frac{1}{r^{k-1}}}$
Let $p_i$ ($0 \leq i \leq k$) be the probability that a vertex in level $k-i$ gets activated.
Hence $p_0=p$ and $p_k$ is the probability of the root being activated in the bootstrap percolation process. We shall write $p_i = h_id^{-\frac{r}{r-1}}$ with $h_0 = h(k)$.  An internal vertex $w$ of the tree becomes activated if it has at least $r$ active children. Hence for $j<k$, using the Poisson approximation $\Pr({\rm Bin}(d,q)=r) \sim e^{-qd}(qd)^r/r!$ we get
$$p_{j+1} \geq \Pr({\rm Bin}(d,p_j) \geq r) \sim e^{-p_jd}(p_jd)^r/r!.$$
As long as $p_j \leq \frac{1}{d}$ then we have that $h_{j+1} \ge \frac{1}{2 \cdot e \cdot r!}(h_j)^r$, and by induction we have that
$$p_{i} \geq (2 e \cdot r!)^{\frac{2}{r}}(\frac{h_0}{(2 e \cdot r!)^{\frac{2}{r}}})^{r^i}d^{-\frac{r}{r-1}} = (2 e \cdot r!)^{\frac{2}{r}}d^{\frac{1}{r^{k-1-i}}}d^{-\frac{r}{r-1}}.$$
Substituting $i = k-1$, children of the root have probability at least $\frac{1}{d}$ to become active, implying that $p_k \ge B$ where $B>0$ is a constant independent of $d$.
The theorem now follows from Corollary~\ref{cor:partialinfection}.
\end{proof}
\medskip

\begin{theorem}
Given an integer $l$, let $G$ be an $(n,d,\lambda)$ graph such $\lambda \le \frac{1}{\sqrt{l}}d$, and $l$ is sufficiently large.
Then $m(G,r)=O(\frac{n}{l^{\frac{r}{r-1}}})$. In particular if $\lambda=O(\sqrt{d})$ then $m(G,r)=O(\frac{n}{d^{\frac{r}{r-1}}})$.
\end{theorem}
\begin{proof}
Follows from Theorem \ref{thm:large_girth_r}, the proof of Theorem \ref{thm:tree_embed}, and Lemma~\ref{lem:general_g,r}.
\end{proof}
\medskip
As in the $r=2$ case, we show our upper bounds are nearly best possible, by analyzing $m(G,r)$ for random $d$-regular graphs.
The following theorem provides lower bounds on $m(G,r)$ when $G$ is sampled according to the configuration model, indicating (in a similar way to the $r=2$ case) that there are $d$-regular graphs for which $\lambda(G) = O(\sqrt{d})$ for which our upper bounds (regarding $m(G,r)$) are nearly tight.
\begin{theorem}

Fix $\epsilon>0$. Then there exist $d_0$ such that for every $d>d_0$ if $G$ is sampled from the configuration model $G^*(n,d)$, then w.h.p. $$ m(G,r) \geq nd^{-(\frac{r}{r-1}+\epsilon)}.$$
\end{theorem}

\begin{proof}
As in Lemma~\ref{lem:subgraph}, if there exists a contagious set of size $t_0$ then for every $t$ such that $t_0 \leq t \leq n$ there is a subgraph of $G$ induced on $t$ vertices, spanning at least $r(t-t_0)$ edges. Set $t_0=nd^{-\frac{r}{r-1}-\epsilon}$, $t=C\cdot t_0$, where $C=C(\epsilon)>2$ will be determined later.
The probability $G$ sampled from $G^*(n,d)$ contains a subgraph of size $t$ spanning at least $r(t-t_0)$ edges is at most
$${n \choose t}{{t \choose 2} \choose r(t-t_0)}\left(\frac{2d}{n}\right)^{r(t-t_0)} \leq $$
$$\left(\frac{en}{t}\right)^t\left(\frac{(Ct_0)^2}{Ct_0}\right)^{r(C-1)t_0} \left(\frac{2d}{n}\right)^{r(t-t_0)}\leq$$
$$\left[\left(\frac{en}{Ct_0}\right)^C\left(\frac{Ct_0 2d}{n}\right)^{r(C-1)}\right]^{t_0}=$$
$$\left[O(1)\frac{t_0^{C(r-1)-r}d^{r(C-1)}}{n^{C(r-1)-r}}\right]^{t_0}.$$
Substituting $t_0=nd^{-\frac{r}{r-1}-\epsilon}$ the above expression simplifies to
$$\left[O(1)d^{r(C-1)-(\frac{r}{r-1}+\epsilon)(C(r-1)-r)}\right]^{t_0}.$$ Taking $C=C(\epsilon)>\frac{r^2}{(r-1)^2\epsilon}$ (we can assume $\epsilon$ is sufficiently small) we get that the probability there exists a contagious set of cardinality smaller than $n \cdot d^{\frac{r}{r-1}+\epsilon}$ is $o(1)$.
\end{proof}
\medskip

\section*{Acknowledgements}
The fourth author would like to thank Boris Pittel for answering questions regarding \cite{BP}, Robert Krauthgamer for his suggestion to study the number of generations until complete activation, and Elchanan Mossel for discussions about bootstrap percolation in random graphs and for referring him to \cite{Janson}.

\begin{appendix}
\section{Hardness of target set selection in regular graphs}
\label{sec:hardness}
We set the activation threshold $r$ to be~2 throughout this section. Recall that it is known that $m(G,2)$, the size of the smallest contagious set, is hard to approximate within any constant factor (and even for factors that depend on $n$)~\cite{Chen09}. The following theorem implies that approximating $m(G,2)$ in regular graphs is roughly as hard as doing so in arbitrary graphs.

\begin{theorem}
\label{thm:regularize}
There is a polynomial time reduction that for every $n$ and every $2 \le \Delta \le n-1$, given an arbitrary graph $G$ with $n$ vertices and maximum degree $\Delta$, transforms $G$ into a $\Delta$-regular graph $H$ on $O(n\Delta^2)$ vertices, such that
$$m(G,2) \le m(H,2) \le 6m(G,2)$$
\end{theorem}

\begin{proof}
Given $\Delta$, we introduce a certain graph that we call a {\em $\Delta$-regularizer}, which will be used as a gadget in our reduction. The $\Delta$-regularizer is a complete graph on $\Delta + 1$ vertices, but with three of its edges removed. The removed edges are picked in such a way that they form a triangle. Hence three vertices, that we call {\em connector vertices}, have degree $\Delta - 2$, and the remaining vertices have degree $\Delta$. Observe that if the three connector vertices are activated, this activates the remaining vertices in the $\Delta$-regularizer. (In fact, when $\Delta \ge 4$, any two vertices are a contagious set for the $\Delta$-regularizer, but this fact is not needed for our proof.)

Given a graph $G(V,E)$ on $n$ vertices and with maximum degree $\Delta$, our reduction works as follows. Make six independent copies of $G$ (with no edges between different copies). Hence now every vertex $v \in V$ has six copies, $v_1, \ldots, v_6$. Let $d_v$ denote the degree of $v$ in $G$. If $d_v < \Delta$, we wish to raise the degrees of each of the vertices of $v_1, \ldots, v_6$ to $\Delta$. To do this we introduce $\Delta - d_v$ fresh copies of the the $\Delta$-regularizer gadget. For every copy of these $\Delta$-regularizers, we introduce edges between its three connector vertices and the six copies of $v$, such that each copy of $v$ gets one new edge, and each connector vertex gets two new edges. Hence all vertices of the $\Delta$-regularizer become of degree $\Delta$, and every copy of $v$ gets $\Delta - d_v$ new edges, making it of degree $\Delta$ as well. Repeating this process for every vertex $u \in V$ (each time with fresh copies of $\Delta$-regularizers) completes the description of the $\Delta$-regular graph $H$.

To see that $m(H,2) \le 6m(G,2)$, consider an arbitrary contagious set in $G$, and observe that taking six copies of this set, one in each copy of $G$, will also activate all of $H$.

To see that $m(G,2) \le m(H,2)$, consider an arbitrary contagious set $S$ in $H$, and observe that the following set $S'$ is contagious in $G$: include vertex $v$ in $S'$ if and only if at least one of its six copies or at least one of the vertices in its $\Delta$-regularizers is in $S$.

Further details are omitted from the proof.
\end{proof}

In the statement and proof of Theorem~\ref{thm:regularize} we preferred simplicity, and hence made no attempt to minimize the size of $H$ or to tighten the relation between $m(G,2)$ and $m(H,2)$.

\section{Contagious sets in non-regular expanding graphs}
\label{sec:irregular}
Our work in this manuscript is concerned with contagious sets in regular graphs, and in nearly regular random graphs. In this section we discuss how insights obtained from these results extend to graphs that are not regular. Rather than attempt to formally define expansion in non-regular graphs (there are several alternative definitions that one may consider), we shall limit our discussion to random graphs (under various models), which would qualify as very good expanders under any reasonable definition of expansion.

Let us set the activation threshold $r$ to be~2 throughout this section. A natural model for random irregular graphs is as follows. Given the number of vertices $n$, one first fixes a degree sequence $d_1 \le d_2 \ldots, \le d_n  \le n-1$, where $\sum_i d_i$ is even. We shall assume that $d_1 \ge 2$, because the activation threshold is~2. Thereafter one draws a multigraph at random using the {\em configuration model} with this degree sequence. Namely, a vertex $i$ corresponds to $d_i$ endpoints of edges, and the multi-graph is generated by selecting a random matching between all endpoints. Thereafter, self loops are removed, and among parallel edges, only one edge is maintained. For degree sequences that will interest us, self loops and parallel edges will be rare and their removal will not significantly change the degree sequence.

Rather than study the configuration model directly, it would be simpler to
consider an alternative process for generating a random non-regular graph, which we illustrate by the following example. Let $d$ be roughly $n^{1/4}$ for concreteness. Generate a random graph $G$ of average degree roughly $d$ using the Erdos-Renyi random graph model $G_{n,p}$ with $p = \frac{d}{n-1}$. By the results of~\cite{JanLuc}, a random subset of $\frac{(1+\delta) n}{2d^2}$ vertices is almost surely contagious. By our Theorem~\ref{thm:LB_random}, the smallest contagious set is of size $\Omega\left(\frac{n}{d^2\log np}\right)$. Now modify $G$ to become a non-regular expander $G'$ as follows: pick at random two disjoint sets of vertices $A$ and $B$, each of size $k = \frac{n}{d^2}$, and within every set, unite all vertices of the set to get a single vertex, thus obtaining vertices $a$ and $b$. Removing parallel edges and self loops that might be generated by this process, each of the vertices $a$ and $b$ has degree roughly $\frac{n}{d}$, whereas the degrees of the remaining vertices remain roughly $d$. In $G'$, the set $\{a,b\}$ is almost surely contagious. (Had we not removed parallel edges, each of $a$ or $b$ by itself would be contagious, and the fact that we take both $a$ and $b$ compensates for the removal of parallel edges. Details are omitted.) Moreover, $a$ and $b$ have multiple common neighbors, and any set of two such common neighbors is contagious as well (because it activates $a$ and $b$).

Returning to the configuration model, the above argument shows that for a degree sequence that has $n-2$ vertices of degree roughly $n^{1/4}$ and two vertices of degree roughly $n^{3/4}$, the size of the smallest contagious set in the corresponding random graph is almost sure the minimum possible, namely, two. Moreover, the contagious set need not contain the high degree vertices. Observe that the average degree $\bar{d}$ of $G'$ is roughly $n^{1/4}$, and hence though an upper bound of $O\left(n/(\bar{d})^2\right)$ on the size of the contagious set holds, this upper bound is very far from being tight.

Let us now modify the degree sequence by scaling all degrees by a factor of $1/\log n$. Namely, there are $n-2$ vertices of degree roughly $n^{1/4}/\log n$ and two vertices of degree roughly $n^{3/4}/\log n$. Observe that for the original nearly regular graph $G$, such a scaling would increase the size of the smallest contagious set by a modest polylogarithmic factor. However, this has a dramatic effect regarding $G'$. The vertices $a$ and $b$ no longer correspond to sets that are sufficiently large to be contagious, and hence the size of the smallest contagious set jumps to at least $\Omega\left(\frac{n}{d^2 \log n}\right) = \Omega(\sqrt{n}\log n)$.

The example above was presented so as to convey two messages.

\begin{itemize}

\item Understanding contagious sets in regular graphs leads us a long way towards understanding contagious sets in irregular graphs. Specifically, in the example above, the non-regular graph $G'$ could be analyzed as a graph derived from a nearly regular graph $G$.

\item Results regarding irregular graphs are much more sensitive to a change in the underlying parameters than the results for regular graphs. Multiplying the degree sequence by a small factor has only a small effect on the size of contagious sets in regular graphs, but a dramatic effect in non-regular graphs. Hence for non-regular graphs, even for random ones, we should not expect to have a single simple parameter (such as average degree) that roughly characterizes the size of contagious sets. This is unlike the case of random nearly regular graphs for which the average degree provides a rough characterization.

\end{itemize}

Another comment that we wish to make is that in certain common models for generating random non-regular graphs, analyzing the size of the smallest contagious set is trivial. Consider the following variation of the well known {\em preferential attachment} model~\cite{Barbasi} with parameter $d \ge 2$. One starts with a clique on $d$ vertices. Thereafter, the remaining vertices arrive one by one in an online fashion. Each new vertex connects to $d$ existing vertices chosen at random, according to some rule that involves the current degrees of existing vertices (e.g., with probability proportional to the degree). Regardless of the rule involved, in such graphs the smallest contagious set is always of size two. Every two of the $d$ initial set of vertices will be contagious (proof by induction on the order of arrival of the vertices).

Further discussion of contagious sets in irregular graphs is beyond the scope of the current paper.
\end{appendix}
\end{document}